\documentclass[a4paper,UKenglish,cleveref, autoref]{lipics-v2019}

\usepackage{hyperref} 
\usepackage{verbatim} 
\usepackage{xcolor} 
\usepackage{mathtools}
\usepackage{amssymb}
\usepackage{amsmath}
\usepackage{amsthm}
\usepackage{stmaryrd}
\usepackage{xspace}
\usepackage{listings}
\usepackage{color}
\usepackage{MnSymbol}

\title{Sound regular corecursion in $\coFJ$ } 

\author{Davide Ancona}{DIBRIS, University of Genova, Italy}{davide.ancona@unige.it}{http://orcid.org/0000-0002-6297-2011}{Member of GNCS (Gruppo Nazionale per il Calcolo Scientifico), INdAM (Istituto Nazionale di Alta Matematica "F. Severi")}
\author{Pietro Barbieri}{DIBRIS, University of Genova, Italy}{pietro.barbieri@edu.unige.it}{https://orcid.org/0000-0003-3193-5549}{}
\author{Francesco Dagnino}{DIBRIS, University of Genova, Italy}{francesco.dagnino@dibris.unige.it}{https://orcid.org/0000-0003-3599-3535}{}
\author{Elena Zucca}{DIBRIS, University of Genova, Italy}{elena.zucca@unige.it}{http://orcid.org/0000-0002-6833-6470}{}

\authorrunning{D. Ancona, P. Barbieri, F. Dagnino, E. Zucca}

\Copyright{Davide Ancona, Pietro Barbieri, Francesco Dagnino, Elena Zucca}

\ccsdesc[500]{Theory of computation~Operational semantics}
\ccsdesc[500]{Software and its engineering~Recursion}
\ccsdesc[500]{Software and its engineering~Semantics}

\keywords{Operational semantics, coinduction, programming paradigms, regular terms}

\category{}

\relatedversion{}

\supplement{}

\nolinenumbers 

\EventEditors{John Q. Open and Joan R. Access}
\EventNoEds{2}
\EventLongTitle{34th European Conference on Object-Oriented Programming}
\EventShortTitle{ECOOP 2020}
\EventAcronym{ECOOP}
\EventYear{2020}
\EventDate{July 13--17, 2020}
\EventLocation{Berlin, Germany}
\EventLogo{}
\SeriesVolume{}
\ArticleNo{}

\newif\ifsubmit
\submittrue

\ifsubmit
\newcommand{\EZ}[1]{{#1}} 

\newcommand{\PB}[1]{{#1}}
\newcommand{\DAComm}[1]{} 
\newcommand{\EZComm}[1]{} 
\newcommand{\FDComm}[1]{}
\newcommand{\PBComm}[1]{}
\else
\newcommand{\EZ}[1]{\textcolor{blue}{#1}} 
 
\newcommand{\PB}[1]{\textcolor{purple}{#1}} 
\newcommand{\DAComm}[1]{{\footnotesize\textcolor{red}{[\bf{Davide: }#1}]}}
\newcommand{\EZComm}[1]{{\footnotesize\textcolor{blue}{[\bf{Elena: }#1}]}}
\newcommand{\FDComm}[1]{{\footnotesize\textcolor{orange}{[\bf{Francesco: }#1}]}}
\newcommand{\PBComm}[1]{{\footnotesize\textcolor{purple}{[\bf{Pietro: }#1}]}}
\fi

\newcommand{\refToFigure}[1]{Figure~\ref{fig:#1}}
\newcommand{\refToSection}[1]{Section~\ref{sect:#1}}
\newcommand{\refToTheorem}[1]{Theorem~\ref{theo:#1}}
\newcommand{\refToLemma}[1]{Lemma~\ref{lemma:#1}}
\newcommand{\refToCor}[1]{Corollary~\ref{cor:#1}} 
\newcommand{\refToProp}[1]{Prop.~\ref{prop:#1}}
\newcommand{\refToPropItem}[2]{\refToProp{#1}(\ref{prop:#1:#2})}
\newcommand{\refToRule}[1]{{\small \textsc{(#1)}}}
\newcommand{\Space}{\hskip 0.8em}
\newcommand{\BigSpace}{\hskip 1.5em}

\newenvironment{proofOf}[1]{\begin{proof}[Proof of #1]}{\end{proof}} 

\newcommand{\fun}[3]{#1:#2\rightarrow#3}
\newcommand{\dom}[1]{\mathit{dom}(#1)}
\newcommand{\callEnvCo}{\callEnv^{\neg\checklabel}}
\newcommand{\callEnvCk}{\callEnv^\checklabel}

\newcommand{\Tuple}[1]    {({#1})}
\newcommand{\Pair}[2]     {\Tuple{{#1},{#2}}}

\newcommand{\FV}[1]{\textit{FV}(#1)}

\newenvironment{grammatica}{$\begin{array}[t]{llll}}{\end{array}$}
\newcommand{\produzione}[3]{#1&{:}{:}=&#2 & \mbox{{\small{#3}}}}
\newcommand{\produzioneCo}[3]{#1&{:}{:}=_{\textsc{co}}&#2 & \mbox{{\small{#3}}}}

\newcommand{\terminale}[1]{\texttt{#1}}
\newcommand{\nonterminale}[1]{\textit{#1}}

\newcommand{\validInd}[2]{#1 \!\vdash\! {#2} }
\newcommand{\valid}[3]{ \Pair{#1}{#2} \!\vdash\! {#3} }
\newcommand{\judg}{\mathit{j}}
\newcommand{\prem}{\textit{Pr}}
\newcommand{\cons}{\textit{c}}
\newcommand{\is}{{\cal I}}
\newcommand{\cois}{{\cal I}^{co}}
\newcommand{\universe}{{\cal U}}
\newcommand{\myrule}{\Rule{\prem}{\cons}}
\newcommand{\Rule}[2]{
\displaystyle                  
\frac{#1}{#2}     
}
\newcommand{\RuleNoName}[3]{\displaystyle\frac{#1}{#2}\ \begin{array}{l} #3 \end{array}}
\newcommand{\thickfrac}[2]{\genfrac{}{}{1.5pt}{}{#1}{#2}}
\newcommand{\CoRule}[2]{\genfrac{}{}{1.7pt}{0}{#1}{#2}}
\newcommand{\NamedRule}[4]{{\tiny\textsc{({#1})}}\displaystyle\frac{#2}{#3}\ \begin{array}{l} #4 \end{array}}
\newcommand{\NamedRuleSimple}[3]{{\tiny\textsc{({#1})}}\displaystyle\frac{#2}{#3}}
\newcommand{\NamedCoRule}[4]{{\tiny\textsc{({#1})}}\displaystyle\thickfrac{#2}{#3}\ \begin{array}{l} #4 \end{array}}
\newcommand{\NamedCoRuleSimple}[3]{{\tiny\textsc{({#1})}}\displaystyle\thickfrac{#2}{#3}}
\newcommand{\Op}[1]{{\textit{F}_{#1}}}
\newcommand{\Ind}[1]{\textit{Ind}(#1)}
\newcommand{\CoInd}[1]{\textit{CoInd}(#1)}
\newcommand{\Generated}[2]{\textit{Gen}(#1,#2)}
\newcommand{\Restricted}[2]{{#1_{{\sqcap}#2}}}
\newcommand{\Spec}{\textit{S}}
\newcommand{\RuleNoFrac}[2]{\genfrac{}{}{0pt}{0}{#1}{#2}}

\newcommand{\FJ}{\ensuremath{\textsc{FJ}\xspace}}
\newcommand{\coFJ}{\ensuremath{\textsc{coFJ}\xspace}}

\newcommand{\CD}{\nonterminale{cd}} 
\newcommand{\CDBar}{\overline{\CD}}
\newcommand{\FDBar}{\overline{\FD}} 
\newcommand{\FD}{\mathit{fd}} 
\newcommand{\f}{\mathit{f}} 

\newcommand{\C}{\mathit{C}} 
\newcommand{\kwThis}{\terminale{this}}
\newcommand{\kwNew}{\terminale{new}}
\newcommand{\kwClass}{\terminale{class}}
\newcommand{\kwExtends}{\terminale{extends}}

\newcommand{\MDBar}{\overline{\MD}}
\newcommand{\MD}{{\nonterminale{md}}}
\newcommand{\m}{\nonterminale{m}}
\newcommand{\x}{\nonterminale{x}} 
\newcommand{\xBar}{\overline{\x}}
\newcommand{\E}{\nonterminale{e}}
\newcommand{\EBar}{\overline{\E}}
\newcommand{\MethBody}[1]{\{#1 \}}

\newcommand{\val}{\nonterminale{v}}
\newcommand{\vBar}{\overline{\val}}

\newcommand{\corec}{\terminale{corec}}
\newcommand{\Any}{\terminale{any}}

\newcommand{\FJvals}{{\cal V}}
\newcommand{\FJexp}{{\cal E}}
\newcommand{\coFJaVals}{\FJvals\textsuperscript{a}}
\newcommand{\coFJaExp}{\FJexp\textsuperscript{a}}
\newcommand{\coFJcVals}{\FJvals\textsuperscript{op}}

\newcommand{\FieldDec}[2]{#1\, #2\texttt{;}}
\newcommand{\NewExpr}[2]{\kwNew\ #1 (#2)}
\newcommand{\FieldAccess}[2]{#1.#2}
\newcommand{\MethodCall}[3]{#1.#2(#3)}

\newcommand{\BigStepFJ}[2]{\eval{#1}{#2}}
\newcommand{\mbody}{\mathit{mbody}}
\newcommand{\fields}{\mathit{fields}}
\newcommand{\Subst}[3]   {#1 [#2/#3]}
\newcommand{\isFJ}{\is_{\scriptscriptstyle \FJ}}

\newcommand{\false}{\texttt{false}}
\newcommand{\true}{\texttt{true}}

\newcommand{\minElem}[2]{\textit{min}(#1,#2)}

\newcommand{\combody}{\textit{co-mbody}}
\newcommand{\eval}[2]{#1\!\Downarrow\!#2}

\newcommand{\coisFJ}{\cois_{\scriptscriptstyle \FJ}}
\newcommand{\evalExtended}[2]{\validInd{\isFJ{\cup}\coisFJ}{\eval{#1}{#2}}}

\newcommand{\sumAll}{\texttt{sum}}

\newcommand{\Caps}[2]{\Pair{#1}{#2}}
\newcommand{\simCaps}{\approx}
\newcommand{\bisim}[3][{}]{ #2 {\simCaps_{#1}} #3} 

\newcommand{\varRel}{\alpha} 
\newcommand{\Quotient}[1]{\mathit{undet}_{\leftrightarrow}(#1)}
\newcommand{\Undetermined}[1]{\mathit{undet}(#1)}
\newcommand{\varEquiv}[3]{ #2 \stackrel{#1}{\leftrightarrow} #3 }

\newcommand{\opsem}[5]{#1,#2,#3\!\Downarrow\!#4,#5}
\newcommand{\intsem}[5]{#1,#2,#3\!\Downarrow_{\textsc{IN}}\!#4,#5}
\newcommand{\open}{{\mathrm{v}}}
\newcommand{\openBar}{\overline{\open}}
\newcommand{\callEnv}{\tau}
\newcommand{\mapEnv}{\sigma}
\newcommand{\resEnv}{\rho}
\newcommand{\callSetRes}{\callSet'}
\newcommand{\callSetResCheck}{{\callSet''}}
\newcommand{\RcallSet}{\EZ{\hat{\callSet}}}
\newcommand{\callSet}{S}
\newcommand{\resEnvOp}[2]{\resEnv^{#1}_{#2}}
\newcommand{\callSetOp}[3]{\callSet^{#1,#2}_{#3}}
\newcommand{\callSetTr}[2]{\callSet^{#1}_{#2}}
\newcommand{\callSetDif}[4]{\callSet^{#1,#2,#3}_{#4}}
\newcommand{\mapEnvPrime}{{\mapEnv'}}
\newcommand{\AppMap}[2]{ #1[#2] }
\newcommand{\unf}{\mathit{unfold}}
\newcommand{\undef}[1]{#1\uparrow} 
\newcommand{\call}{\mathit{c}}
\newcommand{\mapEnvU}{{\widehat{\mapEnv}}}
\newcommand{\envItem}[3]{#1:\!#3^{#2}}
\newcommand{\lbcall}{{t}}
\newcommand{\checklabel}{\textsf{ck}}
\newcommand{\UpdateEnv}[3]{{#1}\{\envItem{#2}{}{#3}\}}
\newcommand{\UpdateEnvCo}[3]{{{#1}\{\envItem{#2}{}{#3}\}}}
\newcommand{\UpdateEnvCk}[3]{{#1}\{\envItem{#2}{{\checklabel}}{#3}\}}
\newcommand{\semCaps}[2]{\llbracket#1,#2\rrbracket}
\newcommand{\expCaps}[2]{\mathit{T}(#1, #2)} 
\newcommand{\mapUnion}[2]{#1{\sqcup} #2}
\newcommand{\bigmapUnion}[2]{\bigsqcup_{#1} #2}
\newcommand{\uptobisim}[2]{#1_{\approx #2}}
\newcommand{\eqclass}[1]{[#1]}

\newcommand{\ApplySubst}[2]{#1\,#2}
\newcommand{\subst}{\theta}
\newcommand{\solutions}{\textsf{Sol}}

\newcommand{\op}{^{\circ}}

\begin{document}

\maketitle

\EZComm{limite per versione finale: 26 pagine esclusa biblio}

\begin{abstract}
{The aim of the paper is to provide solid foundations for a programming paradigm natively supporting the creation and manipulation of cyclic data structures. To this end, we describe $\coFJ$, a Java-like calculus where objects can be \emph{infinite} and methods are equipped with a \emph{codefinition} (an alternative body). 
We provide an abstract semantics of the calculus based on the framework of \emph{inference systems with corules}. In $\coFJ$ with this semantics, $\FJ$ recursive methods on finite objects can be extended to infinite objects as well, and behave as desired by the programmer, by specifying a codefinition. \PB{We also} describe an operational semantics which can be directly implemented in a programming language, and prove the soundness of such semantics with respect to the abstract one.}
\end{abstract}
\section*{Introduction}
Applications often deal with data structures which are conceptually infinite, such as streams or infinite trees. Thus, a major problem for programming languages is how to finitely represent something which is infinite, and, even harder, how to  correctly manipulate such finite representations to reflect the expected behaviour on the infinite structure. 

A well-established solution is \emph{lazy evaluation}, as, e.g., in Haskell. In this approach, the conceptually infinite structure is represented as the result of a function call, which is evaluated only as much as needed.
Focusing on the paradigmatic example of streams (infinite lists) of integers, we can define \texttt{two\_one = 2:1:two\_one}, or even represent the list of natural numbers as \texttt{from 0}, where \texttt{from n = n:from(n+1)}. In this way, functions which only need to inspect a finite portion of the structure, e.g., getting the $i$-th element, can be correctly implemented. On the other hand, functions which need to inspect the whole structure,  e.g., \texttt{min} getting the minimal element, or \texttt{allPos} checking that all elements are positive, \EZComm{ mi pare che in questo approccio non ci sia neanche un equivalente di cose tipo \texttt{incr} perch\`e le funz. ricorsive non sono valori di prima classe} have an undefined result (that is, non-termination, operationally). 

More recently, a different, in a sense complementary\footnote{As we will discuss further in the Conclusion.}, approach has been considered \cite{Jeannin17,AnconaZ12,AnconaDZ17}, which focuses on \emph{cyclic} structures (e.g., cyclic lists, trees and graphs). 
They can be regarded as a particular case of
infinite structures: abstractly, they correspond to regular terms (or trees), that is, finitely branching trees whose depth can be infinite, but contain only a finite set of subtrees. For instance, the list \texttt{two\_one} is regular, whereas the list of natural numbers is not. 
Typically, cyclic data structures are handled by  programming languages by relying on imperative features or ad hoc data structures for bookkeeping. For instance, we can build a cyclic object by assigning to a field of an object a reference to the object itself, or we can visit a graph by marking already encountered nodes.
In this approach \cite{Jeannin17,AnconaZ12,AnconaDZ17}, instead, the programming language natively supports regular structures, as outlined below:
\begin{itemize}
\item Data constructors are enriched by allowing equations, e.g., $\x = 2 : 1 : \x$.
\item Functions are \emph{regularly corecursive}, that is, execution keeps \PB{track} of pending function calls, so that, when the same call is encountered the second time, this is detected, avoiding non-termination as  with ordinary recursion. For instance, when calling \texttt{min} on the list $\x = 2 : 1 : \x$, after an intermediate call on the list $y = 1 : 2 : y$, the same call is encountered.
\end{itemize} 
Regular corecursion originates from \emph{co-SLD resolution} \cite{Simon06,SimonEtAl07,AnconaDovier15}, where already encountered goals
(up to unification), called \emph{coinductive hypotheses}, are considered successful. However, co-SLD resolution is not flexible enough to
to correctly express certain predicates on regular terms; for instance, in the \texttt{min} example,
the intuitively correct  corecursive definition\EZComm{cut: , stating that the minimum of a non-empty list is the minimum between its head and
the minimum of its tail,} is not sound, because the predicate succeeds for all lower bounds of $l$\EZ{, as shown in the following}.   

When moving from goals to functions calls, the same problem manifests more urgently because a result should always be provided for already
encountered calls. To solve this issue, the mechanism of \emph{flexible regular corecursion} can be adopted to allow the programmer to correctly
specify  the behaviour of recursive functions on cyclic structures. For instance, for function \texttt{min}, the programmer specifies that the head of the list should be returned when detecting a cyclic call; in this way, on the list $\x = 2 : 1 : \x$, the result of the cyclic call is $2$, {so that} the result of the original call is $1$, as expected. 

Flexible regular corecursion as outlined above has been proposed in the object-oriented \cite{AnconaZ12}, functional \cite{Jeannin17}, and logic \cite{AnconaDZ17} paradigms (see \refToSection{related} for more details).
However, none of these proposals provides formal arguments for the correctness of  the given operational semantics, by proving
that it is sound with respect to some model of the behaviour of functions (or predicates) on infinite structures. 
The aim of this paper is to bridge this gap, by providing solid foundations for a programming paradigm natively supporting cyclic data structures. This is achieved thanks to the recently introduced framework of \emph{inference systems with corules} \cite{AnconaDZ@ESOP17,Dagnino19}, allowing definitions which are neither inductive, nor purely coinductive. We present the approach in the context of Java-like languages, namely on an extension of Featherweight Java ($\FJ$) \cite{IgarashiEtAl99} called $\coFJ$, outlined as follows:
\begin{itemize}
\item $\FJ$ objects are smoothly generalized from finite to infinite by interpreting their definition coinductively, and methods are equipped with a \emph{codefinition} (an alternative body).
\item We provide an abstract big-step semantics for $\coFJ$ by an inference system with corules. In $\coFJ$ with this semantics, $\FJ$ recursive methods on finite objects can be extended to infinite objects as well, and behave as desired by the programmer, by specifying a codefinition. For instance, if the codefinitions for \texttt{min} and \texttt{allPos} are specified to return the head, and \texttt{true}, respectively, then \texttt{min} returns $1$ on $\x = 2 : 1 : \x$, and $0$ on the list of the natural \PB{numbers}, whereas \texttt{allPos} returns \texttt{true} on both lists. 
\item Then, we provide an operational (hence, executable) semantics where infinite objects are restricted to regular ones and methods are regularly corecursive, and we show that such operational semantics is sound with respect to the abstract one.
\end{itemize}
\EZ{At \url{https://person.dibris.unige.it/zucca-elena/coFJ_implementation.zip}} we provide a prototype implementation of $\coFJ$, \PB{briefly} described in the Conclusion. 
\EZ{A preliminary version of the operational semantics, with no soundness proof with respect to a formal model, has been given in \cite{BDZA19}.} 

\refToSection{corules} is a quick introduction to inference systems with corules. \refToSection{informal} \PB{describes} $\FJ$ and informally introduces our approach. In \refToSection{abstract} we define $\coFJ$ and its abstract semantics, in \refToSection{operational} the operational semantics, \PB{in \refToSection{examples} we show some advanced examples}, and in \refToSection{soundness} we prove soundness. Finally, we discuss related work and draw conclusions in \refToSection{related} and \refToSection{conclu}, respectively.

\section{Inference systems with corules}\label{sect:corules}

First we recall standard notions on inference systems \cite{Aczel77,LeroyGrall09}.
Assuming a \emph{universe} $\universe$ of \emph{judgments}, an \emph{inference system} {$\is$} 
is a set of \emph{(inference) rules}, which are pairs $\myrule$, with $\prem\subseteq\universe$ the set of \emph{premises}, and $\cons\in\universe$ the \emph{consequence} (a.k.a. \emph{conclusion}).  {A rule with an empty set of premises is an \emph{axiom}.} A \emph{proof tree} {(a.k.a. \emph{derivation}) for a judgment $\judg$} is a tree whose nodes are (labeled with) judgments, $\judg$ is the root, and there is a node $c$ with children $\prem$ only if there is a rule $\myrule$. 

The \emph{inductive} and the \emph{coinductive interpretation} of $\is$, denoted $\Ind{\is}$ and $\CoInd{\is}$, are the sets of judgments with, respectively, a finite\footnote{Under the common assumption that sets of premises are finite, otherwise we should say \EZ{well-founded}.}, and a \PB{possibly} infinite proof tree. In set-theoretic terms, let $\fun{\Op{\is}}{\wp(\universe)}{\wp(\universe)}$ be defined by \mbox{$\Op{\is}(S)=\{\cons \mid \prem\subseteq S, \Rule{\prem}{\cons}\in\is\}$}, and say
that a set $S$ is \emph{closed} if $\Op{\is}(S)\subseteq S$, \emph{consistent} if $S\subseteq\Op{\is}(S)$. 
Then,  it can be proved that $\Ind{\is}$ is the smallest closed set, and $\CoInd{\is}$ is the largest consistent set. {We write $\validInd{\is}{\judg}$  when $\judg$ has a finite derivation in $\is$, that is, $\judg\in\Ind{\is}$.


An \emph{inference system with corules}, or \emph{generalized inference system}, is a pair $\Pair{\is}{\cois}$ where $\is$ and $\cois$ are inference systems, whose elements are called \emph{rules} and \emph{corules}, respectively.
Corules can only be used in a special way, as defined below. 

For a subset $S$ of the universe, let $\Restricted{\is}{S}$ denote the inference system obtained from $\is$ by keeping only rules with 
consequence in $S$.
Let $\Pair{\is}{\cois}$ be a generalized inference system. Then, its \emph{interpretation} $\Generated{\is}{\cois}$ is defined by
$\Generated{\is}{\cois}=\CoInd{\Restricted{\is}{\Ind{\is\cup\cois}}}$.
 
In proof-theoretic terms, $\Generated{\is}{\cois}$ is the set of judgments that have a possibly infinite proof tree in $\is$, where all nodes have a finite proof tree in $\is\cup\cois$, that is, the (standard) inference system consisting of rules and corules. 
We write $\valid{\is}{\cois}{\judg}$ when $\judg$ is derivable in $\Pair{\is}{\cois}$, that is, $\judg\in\Generated{\is}{\cois}$. {Note that $\valid{\is}{\emptyset}{\judg}$ is the same as $\validInd{\is}{\judg}$.}

{We illustrate these notions by a simple example. As usual, sets of rules are expressed by \emph{meta-rules} with side conditions, and analogously sets of corules are expressed by \emph{meta-corules} with side conditions.
(Meta-)corules will be written with thicker lines, to be distinguished from (meta-)rules.
The following inference system defines the minimum element of a list, where {$[\x]$ is the list  consisting of only $\x$}, and $\x:u$ the list with head $\x$ and tail $u$.
\[
\Rule{}{\minElem{{[\x]}}{\x}}\BigSpace\Rule{\minElem{u}{y}}{\minElem{x{:}u}{z}}z=\min(x,y) 
\]
The inductive interpretation gives the correct result only on finite lists, since for infinite lists an infinite proof is clearly needed. However, the coinductive one fails to be a function. 
For instance, \PB{for} $L$ the infinite list $2:1:2:1:2:1:\ldots$, any judgment $\minElem{L}{x}$ with $x\leq 1$ can be derived{, as shown below}.
\[
\Rule{
\Rule{
\Rule{\ldots}{\minElem{L}{1}}
}{\minElem{1{:}L}{1}}
}{\minElem{2{:}1{:}L}{1}}
\BigSpace
\Rule{
\Rule{
\Rule{\ldots}{\minElem{L}{0}}
}{\minElem{1{:}L}{0}}
}{\minElem{2{:}1{:}L}{0}}
\]
By adding a corule (in this case a coaxiom), wrong results are ``filtered out'':
\[
\Rule{}{\minElem{x{:}\epsilon}{x}}\BigSpace\Rule{\minElem{u}{y}}{\minElem{x{:}u}{z}}z=\min(x,y) 
\BigSpace\CoRule{}{\minElem{x{:}u}{x}}
\]
Indeed, the judgment $\minElem{2{:}1{:}L}{1}$ has the infinite proof tree shown above, and each node has a finite proof tree in the inference system extended by the corule:
\[
\Rule{
\Rule{
\Rule{\ldots}{\minElem{L}{1}}
}{\minElem{1{:}L}{1}}
}{\minElem{2{:}1{:}L}{1}}
\BigSpace
\Rule{
\CoRule{
}{\minElem{1{:}L}{1}}
}{\minElem{2{:}1{:}L}{1}}
\]
The judgment $\minElem{2{:}1{:}L}{0}$\EZ{, instead,} has the infinite proof tree shown above, but has \emph{no finite proof tree} in the inference system extended by the corule. Indeed, since $0$ does not belong to the list, the corule can never be applied. 
\EZ{On the other hand, the judgment $\minElem{L}{2}$ has a finite proof tree with the corule, but cannot be derived since they it has no infinite proof tree.}
We refer to \cite{AnconaDZ@ESOP17,AnconaDZ@OOPSLA17,AnconaDZ@ECOOP18,Dagnino19} for other examples.}

\EZComm{frase da accorciare ma non so come}
\EZ{As final remark, note that requiring the existence of a finite proof tree with corules only for the root is not enough. For regular proof trees, the requirement to have such a proof tree for each node can be simplified in two ways:
\begin{itemize}
\item either requiring a sufficiently large finite proof-with-corules for the root, that is, a finite proof tree for the root which includes all the nodes of the regular proof tree
\item or requiring a finite proof-with-corules for one node taken from each infinite path.
\end{itemize}}



Let $\Pair{\is}{\cois}$ be a generalized inference system. The \emph{bounded coinduction principle} \PB{\cite{AnconaDZ@ESOP17}}, a generalization of the standard coinduction principle, can be used to prove \emph{completeness} of $\Pair{\is}{\cois}$ w.r.t.\ a set $\Spec$ (for ``specification'') of \emph{valid} judgments. 
\begin{theorem}[Bounded coinduction]\label{theo:bcoind}If the following two conditions hold:
\begin{enumerate}
\item $\Spec \subseteq \Ind{\is\cup\cois}$, that is, each valid judgment has a finite proof tree in $\is\cup\cois$;
\item $\Spec \!\subseteq\! \Op{\is}(\Spec)$, that is, each valid judgment is the consequence of a rule in $\is$ \mbox{with premises in $\Spec$}
\end{enumerate}
then $\Spec\subseteq\Generated{\is}{\cois}$.
\end{theorem}
\EZComm{cut: The standard coinduction principle can be obtained when $\cois=\{\Rule{\emptyset}{\cons}\mid\cons\in\universe\}$; for this particular case the first condition trivially holds.}

\section{From $\FJ$ to $\coFJ$}\label{sect:informal}

We \EZ{recall  $\FJ$}, and informally explain its extension with infinite objects and codefinitions.

\smallskip
\noindent\textbf{Featherweight Java} The standard syntax and semantics in big-step style of $\FJ$  are \PB{shown} in \refToFigure{FJ}. 
\EZ{We omit cast since this feature does not add significant issues.}
We adopt a big-step, rather than a small-step style as in the original $\FJ$ definition, since in this way the semantics is directly defined by an inference system{, denoted $\is_\FJ$ in the following,} which will be equipped with corules to support infinite objects.  
We write $\CDBar$ as metavariable for $\CD_1 \ldots \CD_n$, $n\geq 0$, and analogously for other sequences.
We sometimes use the wildcard $\_$ when the corresponding metavariable is not relevant.

A sequence of class declarations $\CDBar$ is called a \emph{class table}. Each class has a canonical constructor whose parameters match the fields of the class, the inherited \PB{ones} first.
We assume standard $\FJ$ constraints, e.g., no field hiding and no method overloading. {The only variables occurring in method bodies are parameters (including $\kwThis)$.}
Values are \emph{objects}, that is, constructor invocations where arguments are values in turn. 

The judgment $\BigStepFJ{\E}{\val}$ is implicitly parameterized on a fixed class table. In the rules we use standard $\FJ$ auxiliary functions, \EZ{omitting their formal definition.}
Notably, $\fields(\C)$ returns the sequence $\f_1\ldots\f_n$ of the field names\footnote{\EZ{We omit types since not relevant here. We discuss about type systems for $\coFJ$ in the conclusion.}} of the class, in declaration order with the inherited first, and $\mbody(\C,\m)$, for method $\m$ of the class, the pair of the sequence of parameters and the definition. Substitution
$\Subst{\E}{\EBar}{\xBar}$, for $\EBar$ and $\xBar$ of the same length, is defined in the customary manner.
Finally, for $\EBar=\E_1\ldots\E_n$ and $\vBar=\val_1\ldots\val_n$, $\BigStepFJ{\EBar}{\vBar}$ is an abbreviation for $\BigStepFJ{\E_1}{\val_1}\ldots\BigStepFJ{\E_n}{\val_n}$.

\begin{figure}
\hrule
\begin{grammatica}
\\
\produzione{\CD}{\kwClass\
\C\ \kwExtends\ \C'\ \{\ \FDBar\  \MDBar\ \}}{class declaration}\\
\produzione{\FD}{\FieldDec{\C}{\f}}{field declaration}\\
\produzione{\MD}{\C\ \m({\C_1\,\x_1,\ldots,\C_n\,\x_n})\ \MethBody{\E} }{method declaration}\\
\produzione{\E\in\FJexp}{\x\mid\FieldAccess{\E}{\f}\mid\NewExpr{\C}{\EBar}\mid\MethodCall{\E}{\m}{\EBar}}{expression}\\
\\
\produzione{\val\in\FJvals}{\NewExpr{\C}{\vBar}}{{(finite)} object}
\end{grammatica}
\\[2ex]

\hrule

$\begin{array}{l}
\\
\NamedRule{$\FJ$-field}{\BigStepFJ{\E}{\val}}{\BigStepFJ{\FieldAccess{\E}{\f}}{\val_i}}
{\begin{array}{l}
\val=\NewExpr{\C}{{\val_1,\ldots,\val_n}}\\
\fields(\C)={\f_1\ldots\f_n}\\
\f=\f_i, i\in 1..n
\end{array}}
\BigSpace
\NamedRule{$\FJ$-new}{\BigStepFJ{\EBar}{\vBar}}{\BigStepFJ{\NewExpr{\C}{\EBar}}{\NewExpr{\C}{\vBar}}}
{
}
\\[6ex]
\NamedRule{$\FJ$-invk}{\BigStepFJ{\E_0}{\val_0}\Space\BigStepFJ{\EBar}{\vBar}\Space\BigStepFJ{\Subst{\Subst{\E}{\val_0}{\kwThis}}{\vBar}{\xBar}}{\val}}{\BigStepFJ{\MethodCall{\E_0}{\m}{\EBar}}{\val}}
{\begin{array}{l}
\val_0=\NewExpr{\C}{\_}\\
\mbody(\C,\m)=\Pair{\xBar}{\E}
\end{array}}
\\[2ex]
\end{array}$
\\[2ex]

%
\hrule
\caption{$\FJ$ syntax and big-step rules}\label{fig:FJ}
\end{figure}

Rule \refToRule{{FJ-field}} models field access. 
If the selected field is actually a field of the receiver's class, then the corresponding value is returned as result.
Rule \refToRule{FJ-new} models object creation: if the {argument expressions} $\EBar$ evaluate to values $\vBar$, {then} the result is an object of {class} \textit{C}. 
\EZComm{cut: The rule is faithful to the original semantics of $\FJ$ \cite{IgarashiEtAl99}: no check is performed on the existence of the definition
of \textit{C} and on the number of arguments passed to the constructor.}
Rule \refToRule{FJ-invk} models method invocation. The receiver and argument expressions are evaluated first. Then, method look-up is performed, starting from the receiver's class,  by the auxiliary function $\mbody$. Lastly, the definition $\E$ of the method, where $\kwThis$ is replaced by the receiver, and the parameters by the arguments, is evaluated, and its result is returned.

\smallskip
\noindent\textbf{Infinite objects and codefinitions}
We take as running example the following $\FJ$ implementation of lists of integers, equipped with some typical methods: {\texttt{isEmpty} tests the emptiness,} \texttt{incr} returns the list where all elements have been incremented by one, \texttt{allPos} checks whether all elements are positive, \texttt{member} checks whether the argument is in the list, and \texttt{min} returns the minimal element.
\begin{lstlisting}
class List extends Object {
  bool isEmpty() {true}
  List incr() {new EmptyList()}
  bool allPos() {true}
  bool member(int x) {false}
}
class EmptyList extends List { }
class NonEmptyList extends List { 
  int head; List tail;
  bool isEmpty() {false}
  List incr() {new NonEmptyList(this.head+1,this.tail.incr())}
  bool allPos() {if (this.head<=0) false else this.tail.allPos()}
  bool member(int x) {if (this.head==x) true else this.tail.member(x)}
  int min() {
    if (this.tail.isEmpty()) this.head
    else Math.min(this.tail.min(),this.head)
  }
}
\end{lstlisting}
We used some additional standard constructs, such as conditional and primitive types \texttt{bool} and \texttt{int} with their operations; to avoid to use abstract methods, \texttt{List} provides the default implementation on empty lists, overridden in \texttt{NonEmptyList}, except for method \texttt{min} which is only defined on non empty lists.

In $\FJ$ we can represent finite lists. 
For instance, the object
\begin{lstlisting}
new NonEmptyList(2, new NonEmptyList(1, new EmptyList()))
 \end{lstlisting}
which we will abbreviate $[2,1]$, represents a list of two elements, and it is easy to see that all the above method definitions provide the expected meaning on finite lists.

On the other  hand, since the syntactic definition for objects is interpreted, \PB{like} the others, inductively, in $\FJ$ objects are \emph{finite}, hence we cannot represent, e.g., the infinite list of natural numbers $[0,1,2,3, \ldots]$, abbreviated $[0..]$, or the infinite list $[2,1,2,1,2,1,\ldots]$, abbreviated $[2,1]^\omega$. 
To move from finite to infinite objects, it is enough to interpret the syntactic definition for values \emph{coinductively}, so to obtain infinite terms as well.
However, to make the extension significant, we should be able to \emph{generate} such infinite objects as results of expressions, and to appropriately \emph{handle} them by methods.

To generate infinite objects, e.g., the infinite lists mentioned above,  a natural approach is to consider method definitions as \emph{corecursive}, that is, to take the \emph{coinductive} interpretation of the inference system in \refToFigure{FJ}. 
Consider the following class:
\begin{lstlisting}
class ListFactory extends Object {
  NonEmptyList from(int x) {new NonEmptyList(x, this.from(x+1)}
  NonEmptyList two_one() {new NonEmptyList(2, this.one_two())}
  NonEmptyList one_two() {new NonEmptyList(1, this.two_one())}
}
\end{lstlisting}
With the standard $\FJ$ semantics, given by the inductive interpretation of the inference system in \refToFigure{FJ}, the method invocation \texttt{new ListFactory().from(0)} (abbreviated $\texttt{from}_0$ in the following) has no result, since there is no finite proof tree for a judgment of shape $\eval{\texttt{from}_0}{\_}$. 
Taking the coinductive interpretation, instead, such call returns as result the infinite list of natural numbers $[0..]$, since there is an infinite proof tree for the judgment $\eval{\texttt{from}_0}{[0..]}$.
Analogously, the method invocation $\texttt{new ListFactory().two\_one()}$ returns $[2,1]^\omega$.
Moreover, the method invocations $\MethodCall{[0..]}{\texttt{incr}}{}$ and $\MethodCall{[2,1]^\omega}{\texttt{incr}}$ correctly return as result the infinite lists $[1..]$ and $[3,2]^\omega$, respectively.

However, in many cases to consider method definitions as corecursive is not satisfactory, since it leads to non-determinism\EZ{, as  shown for inference systems in \refToSection{corules}}. For instance, for the method invocation  $\MethodCall{[0..]}{\texttt{allPos}}{}$ both judgments $\eval{\MethodCall{[0..]}{\texttt{allPos}}{}}{\true}$ and $\eval{\MethodCall{[0..]}{\texttt{allPos}}{}}{\false}$ are derivable, and analogously for $\MethodCall{[2,1]^\omega}{\texttt{allPos}}{}$. In general, both results can be obtained for any infinite list of all positive numbers. 
A similar behavior is exhibited by method \texttt{member}: given an infinite list $L$ which does not contain $\x$, both judgments $\eval{\MethodCall{L}{\texttt{member}}{\x}}{\true}$ and $\eval{\MethodCall{L}{\texttt{member}}{\x}}{\false}$ are derivable.
Finally, for the method invocation $\MethodCall{[2,1]^\omega}{\texttt{min}}{}$, any judgment $\eval{\MethodCall{[2,1]^\omega}{\texttt{min}}{}}{\x}$ with $\x\leq 1$ can be derived. 

To solve this problem, $\coFJ$ allows the programmer to \emph{control} the semantics of corecursive methods by adding a \emph{codefinition}\footnote{\EZ{The term ``codefinition'' is meant to suggest ``alternative definition used to handle corecursion''.}}, that is, an alternative method body playing a special role. 
Depending on the codefinition, the purely coinductive interpretation is refined, by filtering out some judgments. In the example, to achieve the expected meaning, the programmer should provide the following codefinitions.

\begin{lstlisting}
class ListFactory extends Object {
  NonEmptyList from(int x) {
    new NonEmptyList(x, this.from(x+1)} corec {any}
  NonEmptyList one_two() {
    new NonEmptyList(1, this.two_one())} corec {any}
  NonEmptyList two_one() {
    new NonEmptyList(2, this.one_two())} corec {any}
}
class NonEmptyList extends List { 
  int head; List tail;
  bool isEmpty() {false}
  List incr() { 
    new NonEmptyList(this.head+1,this.tail.incr())} corec {any}
  bool allPos() { 
    if (this.head <= 0) false else this.tail.allPos()} corec {true}
  bool member(int x) { 
    if (this.head == x) true else this.tail.member(x)} corec {false} 
  int min() {
    if (this.tail.isEmpty()) this.head
    else Math.min(this.tail.min(),this.head) 
  } corec {this.head}
}
\end{lstlisting}
For the three methods of \texttt{ListFactory} and for the method \texttt{incr} the codefinition is \texttt{any}. This corresponds to \PB{keeping} the coinductive interpretation as it is, as \PB{is} appropriate in these cases since it provides only the expected result. In the other three methods, instead, the effect of the codefinition is to filter the results obtained by the coinductive \EZ{interpretation}. The way this is achieved is explained in the following section. {Finally, for method \texttt{isEmpty} no codefinition is added, since the inductive behaviour works on infinite lists as well.}\EZComm{stessa cosa per le versioni di default dei metodi, dire qualcosa? Mi viene anche in mente che in caso di inheritance dovremmo indagare se ci sono vincoli sulla ridefinizione della codefinizione}

\section{$\coFJ$ and its abstract semantics}\label{sect:abstract}
We formally define $\coFJ$, illustrate how the previous examples get the expected semantics, and show that, despite its non-determinism, $\coFJ$ is a conservative extension of $\FJ$.

\smallskip
\noindent\textbf{Formal definition of $\coFJ$} The $\coFJ$ syntax is given in \refToFigure{syntax}. 
\begin{figure}
\begin{small}
\begin{grammatica}
\produzione{\CD}{\kwClass\
\C\ \kwExtends\ \C'\ \{\ \FDBar\  \MDBar\ \}}{class declaration}\\
\produzione{\FD}{\FieldDec{\C}{\f}}{field declaration}\\
\produzione{\MD}{\C\ \m({\C_1\,\x_1,\ldots,\C_n\,\x_n})\ \MethBody{\E}\ [\corec\ \MethBody{\E'}]}{method declaration with codefinition}\\
\produzione{\E\in\FJexp}{\x\mid\FieldAccess{\E}{\f}\mid{\NewExpr{\C}{\EBar}\mid\MethodCall{\E}{\m}{\EBar}}}{expression}\\[2ex]
\produzioneCo{\val\in\coFJaVals}{\NewExpr{\C}{\vBar}}{{possibly infinite object}}\\
\produzione{\E\in\coFJaExp}{\x\mid\FieldAccess{\E}{\f}\mid{\NewExpr{\C}{\EBar}\mid\MethodCall{\E}{\m}{\EBar}}\mid\val}{{runtime expression}}
\end{grammatica}
\\[2ex]

\hrule

$\begin{array}{l}
\\
\NamedRule{abs-field}{\eval{\E}{\val}}{\eval{\FieldAccess{\E}{\f}}{\val_i}}
{\begin{array}{l}
\val=\NewExpr{\C}{{\val_1,\ldots,\val_n}}\\
\fields(\C)=\PB{\f_1...\f_n}\\
\f=\f_i, i\in 1..n
\end{array}
}
\BigSpace
\NamedRule{abs-new}{\eval{\EBar}{\vBar}}{\eval{\NewExpr{\C}{\EBar}}{\NewExpr{\C}{\vBar}}}
{}
\\[6ex]
\NamedRule{abs-invk}{\eval{\E_0}{\val_0}\Space\eval{\EBar}{\vBar}\Space\eval{\Subst{\Subst{\E}{\val_0}{\kwThis}}{\vBar}{\xBar}}{\val}}{\eval{\MethodCall{\E_0}{\m}{\EBar}}{\val}}
{
\val_0=\NewExpr{\C}{\_}\\
{\mbody(\C,\m)=\Pair{\xBar}{\E}}
}
\BigSpace
\NamedCoRule{{abs-co-val}}{}{\eval{\val}{\val}}
{}
\\[6ex]
\NamedCoRule{abs-co-invk}{\eval{\E_0}{\val_0}\Space\eval{\EBar}{\vBar}\Space\eval{\Subst{\Subst{\Subst{\E'}{\val_0}{\kwThis}}{\vBar}{\xBar}}{\val}{\Any}}{\val_\textit{co}}}{\eval{\MethodCall{\E_0}{\m}{\EBar}}{\val_\textit{co}}}
{
\val_0=\NewExpr{\C}{\_}\\
{\combody(\C,\m)=\Pair{\xBar}{\E'}}
}
\end{array}$
\end{small}
\caption{$\coFJ$ syntax and abstract semantics}\label{fig:syntax}
\end{figure}
As the reader can note, the only difference is that method declarations include now, besides a definition $\E$, an optional \emph{codefinition} $\E'${, as denoted by the square brackets in the production}. Furthermore, besides $\kwThis$, there is another special variable $\Any${, which can only occur in codefinitions}.
The codefinition will be used to provide an abstract semantics through an inference system with corules, where the role of $\Any$ is to be a placeholder for an arbitrary value.
For simplicity, \PB{we require the codefinition $\E'$ to be} statically restricted to avoid recursive (even indirect) calls to the same method \EZ{(we omit the standard formalization).}
{Note that $\FJ$ is a (proper) subset of $\coFJ$: indeed, an $\FJ$ class table is a $\coFJ$ class table with no codefinitions}. 

The syntactic definition for values is the same as before, but is now interpreted \emph{coinductively}, as indicated by the symbol $::=_\textsc{co}$. In this way, infinite objects are supported. 
By replacing method parameters by arguments, we obtain \emph{runtime expressions} admitting infinite objects as subterms.  
The sets $\FJvals$ and $\FJexp$ of $\FJ$ objects and expressions are subsets of $\coFJaVals$ and $\coFJaExp$, respectively. The judgment $\eval{\E}{\val}$, with $\E\in\coFJaExp$ and $\val\in\coFJaVals$,  is defined by an inference system with corules {$\Pair{\isFJ}{\coisFJ}$ where the rules $\isFJ$ are those\footnote{To be precise, meta-rules are the same, with meta-variables $\E$ and $\val$ ranging on $\coFJaExp$, and $\coFJaVals$, respectively. However, we could have taken this larger universe in $\FJ$ as well without affecting the defined relation.} of $\FJ$, as in \refToFigure{FJ}, and the corules $\coisFJ$ are instances of two metacorules.}

Corule \refToRule{{abs-co-val}} is needed to obtain a value for infinite objects\EZ{, as shown below.}
Corule \refToRule{abs-co-invk} is analogous to the standard rule for method invocation, but uses the codefinition, and  the variable $\Any$ can be non-deterministically substituted with an arbitrary value. The auxiliary function \textit{co-mbody} is defined analogously to \textit{mbody}, but it returns the codefinition. Note that, even when $\mbody(\C,\m)$ is defined,  $\combody(\C,\m)$ can be undefined since no codefinition has been specified. This can be done to force a purely inductive behaviour for the method.

\smallskip
\noindent\textbf{Examples} \EZComm{ho aggiustato un po' la figura rendendo pi\`u uniformi le notazioni}
As an example, we illustrate in \refToFigure{example_lf} the role of the two corules for the call \texttt{new ListFactory().from(0)}.
 For brevity, we write abbreviated class names. 
 Furthermore, $\texttt{from}_n$ stands for the call \texttt{new ListFactory().from(n)} and $[n..]$ for the infinite object \texttt{new NonEmptyList(n,new NonEmptyList(n+1,...)))}.

\begin{figure}
\begin{footnotesize}
\begin{math}
\EZ{T_n {=} }\NamedRuleSimple{abs-invk}{\NamedRuleSimple{abs-new}{}{\eval{\texttt{new LF()}}{\texttt{new LF()}}} \quad
\NamedRuleSimple{n-val}{}{\eval{n}{n}}\quad
\NamedRuleSimple{abs-new}{\NamedRuleSimple{n-val}{}{\eval{n}{n}}\qquad \RuleNoFrac{}{T_n}
}{\eval{\texttt{new NEL($n$,}\texttt{new LF().from($n$+1)})}{[n..]}}
}{\eval{\texttt{from}_n}{[n..]}}
\\[4ex]
T_{\EZ{n+1}} {=} \NamedRuleSimple{abs-invk}{\NamedRuleSimple{abs-new}{}{\eval{\texttt{new LF()}}{\texttt{new LF()}}} \quad
  \NamedRuleSimple{+}
  {\cdots}{\eval{\texttt{$n$+1}}{n{+}1}}\quad
{\NamedRuleSimple{abs-new}{\NamedRuleSimple{n-val}{}{\eval{n{+}1}{n{+}1}}\quad \RuleNoFrac{}{T_{\EZ{n+2}}}
}{\scriptsize \eval{\texttt{new NEL($n{+}1$,}\texttt{new LF().from($n{+}1$+1)})}{[n{+1}..]}}}
}{\eval{\texttt{new LF().from($n$+1)}}{[n{+}1..]}}
\\[10ex]
\NamedCoRule{abs-co-invk}{\NamedRule{abs-new}{}{\eval{\texttt{new LF()}}{\texttt{new LF()}}}{}\Space{\NamedRule{n-val}{}{\eval{n}{n}}{}}\Space\NamedCoRule{abs-co-val}{}{[n..]\equiv\eval{\Subst{\Subst{\Any}{{\texttt{new LF()}}}{\kwThis}}{[n..]}{\Any}}{[n..]}}{}}{\eval{\texttt{from}_n}{[n..]}}{}
\end{math}
\end{footnotesize}
\caption{Infinite (top) and finite (bottom) proof trees for $\eval{\texttt{from}_n}{[n..]}$}\label{fig:example_lf}
\end{figure}

In the top part of \refToFigure{example_lf}, we show the infinite proof tree \EZ{$T_n$} which can be constructed, for any natural number $n$,
for the judgment $\eval{\texttt{from}_n}{[n..]}$ without
the use of corules. We use standard rules \refToRule{n-val} and \refToRule{+} to deal with integer constants and addition.

To derive the judgment in the inference system with corules, each node in this infinite tree should have a finite proof tree \EZ{with} the corules.
Notably, this should hold  for nodes of shape $\eval{\texttt{from}_n}{[n..]}$, and indeed the finite proof tree for such nodes is shown in the bottom part of the figure.
Note that, in this example, the result for the call $\texttt{from}_n$
is uniquely determined by the rules, hence the role of the corules is just to ``validate'' this result. To this end, the codefinition of the method \texttt{from} is the special variable \texttt{any}, which, when evaluating the codefinition, can be replaced by any value, hence, in particular, by the correct result $[n..]$. Corule \refToRule{abs-co-val}  is needed to obtain a finite proof tree for the infinite objects of shape $[n..]$. Analogous infinite and finite proof trees can be constructed for the judgments $\eval{\texttt{new ListFactory().two\_one()}}{[2,1]^\omega}$, $\eval{\MethodCall{[0..]}{\texttt{incr}}{}}{[1..]}$ and $\eval{\MethodCall{[2,1]^\omega}{\texttt{incr}}}{[3,2]^\omega}$. 

For the method call $\MethodCall{[0..]}{\texttt{allPos}}{}$, instead, both judgments $\eval{\MethodCall{[0..]}{\texttt{allPos}}{}}{\true}$ and $\eval{\MethodCall{[0..]}{\texttt{allPos}}{}}{\false}$ have an infinite proof tree. However, no finite proof tree using the codefinition can be constructed for the latter, whereas this is trivially possible for the former. Analogously, given an infinite list $L$ which does not contain $\x$, only the judgment $\eval{\MethodCall{L}{\texttt{member}}{\x}}{\false}$ has a finite proof tree using the codefinition. 

Finally, for the method invocation $\MethodCall{[2,1]^\omega}{\texttt{min}}{}$, for any $\val\leq 1$ there is an infinite proof tree
built without corules for the judgment $\eval{\mathtt{[2,1]^\omega .min()}}{\val}$ as shown in \refToFigure{example_min_inf}.
\begin{figure}[h]
\begin{footnotesize}
$\NamedRuleSimple{abs-invk}{T_0 \qquad T_1}{\eval{\mathtt{[2,1]^\omega .min()}}{\val}} 
$
\qquad
$
T_0{=} \NamedRuleSimple{abs-new}{\NamedRuleSimple{n-val}{}{\eval{\mathtt{2}}{\mathtt{2}}}\qquad\NamedRuleSimple{abs-new}{\NamedRuleSimple{n-val}{}{\eval{\mathtt{1}}{\mathtt{1}}}\qquad \RuleNoFrac{}{T_0}}{\eval{\mathtt{[1,2]^\omega }}{\mathtt{[1,2]^\omega }}}}{\eval{\mathtt{[2,1]^\omega }}{\mathtt{[2,1]^\omega }}}
$
\\[3ex]
$T_1 {=}  \NamedRuleSimple{\EZ{if-f}}{\Rule{\vdots}{\eval{\mathtt{[2,1]^\omega .tail.isEmpty()}}{\mathtt{false}}}\quad \Rule{\Rule{\RuleNoFrac{T_2}{\vdots}}{\eval{\mathtt{[2,1]^\omega .tail.min()}}{\val}}\quad\Rule{\vdots}{\eval{\mathtt{[2,1]^\omega .head}}{\mathtt{2}}}} {\eval{\mathtt{Math.min([2,1]^\omega .tail.min(),[2,1]^\omega .head)}}{\val}}}{\eval{\mathtt{\mathbf{if}\ [2,1]^\omega .tail.isEmpty()\ \mathbf{then}\ [2,1]^\omega .head\ \mathbf{else}\ Math.min([2,1]^\omega .tail.min(),[2,1]^\omega .head)}}{\val}}
$
\\[3ex]
$T_2 {=}  \NamedRuleSimple{\EZ{if-f}}{\Rule{\vdots}{\eval{\mathtt{[1,2]^\omega .tail.isEmpty()}}{\mathtt{false}}}\quad \Rule{\Rule{\RuleNoFrac{T_1}{\vdots}}{\eval{\mathtt{[1,2]^\omega .tail.min()}}{\val}}\quad\Rule{\vdots}{\eval{\mathtt{[1,2]^\omega .head}}{\mathtt{1}}}} {\eval{\mathtt{Math.min([1,2]^\omega .tail.min(),[1,2]^\omega .head)}}{\val}}}{\eval{\mathtt{\mathbf{if}\ [1,2]^\omega .tail.isEmpty()\ \mathbf{then}\ [1,2]^\omega .head\ \mathbf{else}\ Math.min([1,2]^\omega .tail.min(),[1,2]^\omega .head)}}{\val}}
$
\end{footnotesize}
\caption{Infinite proof tree for $\eval{\mathtt{[2,1]^\omega .min()}}{\val}$ with $\val\leq 1$ (main tree at the top left corner)}\label{fig:example_min_inf}
\end{figure}
However, only the judgment $\eval{\MethodCall{[2,1]^\omega}{\texttt{min}}{}}{1}$  has a finite proof tree using the codefinition (\refToFigure{example_min_fin}). 
For space reasons in both figures ellipses are used to omit the less interesting parts of the proof trees;
we use the standard rule \refToRule{\EZ{if-f}} for conditional, and the predefined function \texttt{Math.min} on integers.

\begin{figure}[h]
\begin{footnotesize}
$
\NamedRule{abs-invk}{\RuleNoFrac{}{T_0}\ \NamedRule{\EZ{if-f}}{\Rule{\vdots}{\eval{\mathtt{[2,1]^\omega .tail.isEmpty()}}{\mathtt{false}}}{}\quad\Rule{\NamedCoRuleSimple{abs-co-invk}{\cdots\quad\Rule{\NamedCoRuleSimple{abs-co-val}{}{\eval{\mathtt{[1,2]^\omega }}{\mathtt{[1,2]^\omega }}}}{\eval{\mathtt{[1,2]^\omega .head}}{\mathtt{1}}}}{\eval{\mathtt{[2,1]^\omega .tail.min()}}{\mathtt{1}}}\quad\Rule{\vdots}{\eval{\mathtt{[2,1]^\omega .head}}{\mathtt{2}}}}{\eval{\mathtt{Math.min([2,1]^\omega .tail.min(),[2,1]^\omega .head)}}{\mathtt{1}}}}{\eval{\mathtt{\mathbf{if}\ [2,1]^\omega .tail.isEmpty()\ \mathbf{then}\ [2,1]^\omega .head\ \mathbf{else}\ Math.min([2,1]^\omega .tail.min(),[2,1]^\omega .head)}}{\mathtt{1}}}{}}{\eval{\mathtt{[2,1]^\omega .min()}}{\mathtt{1}}}{}
$
\end{footnotesize}
\caption{Finite proof tree with codefinition for $\eval{\mathtt{[2,1]^\omega .min()}}{\mathtt{1}}$ ($T_0$ as in \refToFigure{example_min_inf})}\label{fig:example_min_fin}
\end{figure}

\smallskip
\noindent\textbf{Non-determinism and conservativity} The $\coFJ$ abstract semantics is inherently non-deterministic. Indeed, depending on the codefinition, the non-determinism of the coinductive interpretation may be kept. For instance, consider the following method declaration:
\begin{lstlisting}
class C { 
  C m() { this.m() } corec { any }
}
\end{lstlisting}
Method \texttt{m()} recursively calls itself. In the abstract semantics, the judgment $\eval{\MethodCall{\NewExpr{\texttt{C}}{}}{\texttt{m}}{}}{\val}$ can be derived for any value $\val$. In the operational semantics defined in \refToSection{operational}, such method call evaluates to {$\Caps{\x}{\x:\x}$}, that is, the representation of \emph{undetermined}.

{However, determinism of $\FJ$ evaluation is preserved. Indeed, $\coFJ$ abstract semantics is a \emph{conservative} extension of $\FJ$ semantics, as formally stated below.}

\begin{theorem}[Conservativity] \label{theo:abs-conservative} 
If  $\validInd{\isFJ}{\BigStepFJ{\E}{\val}}$, then
\PB{$\valid{\isFJ}{\coisFJ}{\eval{\E}{\val'}}$ \EZ{iff} $\val = \val'$.}
\end{theorem}
\PB{\begin{proof}
Both \EZ{directions} can be easily proved by induction on the definition of $\validInd{\isFJ}{\BigStepFJ{\E}{\val}}$. 
For \EZ{the left-to-right direction}, the fact that each syntactic category has a unique applicable meta-rule is crucial.
\end{proof}}

This theorem states that, {whichever the codefinitions} chosen, $\coFJ$ does not change the semantics of {expressions} evaluating to some value in $\FJ$.
That is, $\coFJ$ abstract semantics allows derivation of new values only for {expressions} whose semantics is undefined in standard $\FJ${, as in the examples shown above.} 
Note also that, if no codefinition is specified, then the $\coFJ$ abstract semantics \emph{coincides} with the $\FJ$ one,
because corule  \refToRule{abs-co-invk} cannot be applied, hence no infinite proof trees can be built
for the evaluation of $\FJ$ {expressions}.

\section{Operational semantics}\label{sect:operational}
We informally introduce the operational semantics of $\coFJ$, provide its formal definition, and prove that it is deterministic and conservative.

\smallskip
\noindent\textbf{Outline} \PB{In contrast to} the abstract semantics of the previous section, {the aim is to define a semantics which leads to an interpreter for the calculus}. To obtain this, there are two issues to be considered:
\begin{enumerate}
\item infinite (regular) objects should be represented in a finite way;
\item infinite (regular) proof trees should be replaced by finite proof trees.
\end{enumerate}
In the following we explain how these issues are handled in the $\coFJ$ {operational} semantics.

To obtain (1), we use an approach based on \emph{capsules} \cite{JeanninK12}, which are essentially expressions supporting cyclic references. In our context, capsules are pairs $\Caps{\E}{\mapEnv}$ where $\E$ is an $\FJ$ expression and $\mapEnv$ is an \emph{environment}, that is, a finite mapping from variables into $\FJ$ expressions. Moreover, the \EZ{following} \emph{capsule property} is satisfied: \EZ{writing} $\FV{\E}$ \EZ{for} the set of free variables in $\E$, $\FV{\E}\subseteq\dom{\mapEnv}$ and, for all $\x\in\dom{\mapEnv}$, $\FV{\mapEnv(\x)}\subseteq\dom{\mapEnv}$. An $\FJ$ source expression $\E$ is represented by the capsule $\Caps{\E}{\emptyset}$, where $\emptyset$ denotes the empty {environment}. In particular, values are pairs $\Caps{\open}{\mapEnv}$ where $\open$ is an \emph{open} $\FJ$ object, that is, an object possibly containing variables. In this way, cyclic objects can be obtained: for instance, {$\Caps{\x}{\x:\NewExpr{\texttt{NEL}}{2,\NewExpr{\texttt{NEL}}{1,\x}}}$ represents the infinite regular list $[2,1]^\omega$} considered before.

To obtain (2), methods are \emph{regularly corecursive}. This means that execution keeps track of the pending method calls, so that, when a call is encountered the second time, this is detected\footnote{\PB{The semantics detects an already encountered call by relying on capsule equivalence (\refToFigure{notations}).}}, avoiding non-termination as it would happen with ordinary recursion. Regular corecursion in $\coFJ$ is \emph{flexible}, since the behaviour of the method when a cycle is detected is specified by the codefinition. 

Consider, for instance,  the method call {$\texttt{new ListFactory().two\_one()}$; thanks to regular corecursion, the result is the cyclic object
$\Caps{\x}{\x:\NewExpr{\texttt{NEL}}{2,\NewExpr{\texttt{NEL}}{1,\x}}}$}. Indeed, the operational semantics associates \EZ{a fresh variable, say,  $\x$, to the initial call, so that,} when the same call is encountered the second time,  the association $\x:{\x}$ is added in the environment, and the codefinition is evaluated where \texttt{any} is replaced by $\x$. Hence, $\Caps{\x}{\x:\x}$ is returned  as result, so that the result of the original call is {$\Caps{\x}{\x:\NewExpr{\texttt{NEL}}{2,\NewExpr{\texttt{NEL}}{1,\x}}}$}. \EZ{The call \texttt{new ListFactory().from(0)}, instead,} does not terminate in the operational semantics, since no call is encountered more than once (the resulting infinite object is non-regular).

Consider now the call {$\MethodCall{[2,1]^\omega}{\texttt{allPos}}{}$}.
In this case, when the call is encountered the second time, after an intermediate call {$\MethodCall{[1,2]^\omega}{\texttt{allPos}}{}$}, the result of the evaluation of the codefinition is \texttt{true}, so that the result of the original call is \texttt{true} as well.\footnote{To be rigorous, a capsule of shape $\Caps{\true}{\_}$.} If the codefinition were \texttt{any}, then the result would be $\Caps{\x}{\x:\x}$, that is, undetermined.
Note that, if the list is finite, then no regular corecursion is involved, since the same call cannot occur more {than} once; the same holds if the list is cyclic, but contains
a non-positive element, hence the method invocation returns \lstinline{false}.
The only case requiring regular corecursion is when the method is invoked on a cyclic list with 
all positive elements, as {$[2,1]^\omega$}. 

{In the case of $\MethodCall{[2,1]^\omega}{\texttt{min}}{}$, when the call is encountered the second time the result of the evaluation of the codefinition is $2$, so that the result of the intermediate call $\MethodCall{[1,2]^\omega}{\texttt{min}}{}$ is $1$, and this is also the result of the original call.}

\smallskip
\noindent\textbf{Formal definition} To formally express the approach described above, the judgment of the operational semantics has shape $\opsem{\E}{\mapEnv}{\callEnv}{\open}{\mapEnvPrime}$ where: $\Caps{\E}{\mapEnv}$ is the capsule to be evaluated; $\callEnv$ is a \emph{call trace}, used to keep track of already encountered calls, that is, {an injective} map from \emph{calls} $\MethodCall{\open_0}{\m}{\openBar}$ to {(possibly tagged)} variables, and $\Caps{\open}{\mapEnvPrime}$ is the capsule result. {Variables in the codomain of the call trace {have a tag} $\checklabel$ during the checking step for the corresponding call, as detailed below.} The pair $\Caps{\E}{\mapEnv}$ and $\Caps{\open}{\mapEnvPrime}$ are assumed to satisfy the capsule property.  

\PB{The semantic rules} are given in \refToFigure{new-sem}. We denote by $\UpdateEnv{\mapEnv}{\x}{\open}$ the environment which gives $\open$ on $\x$, and is equal to $\mapEnv$ elsewhere{, and analogously for other maps}. Furthermore, we use the following notations, formally defined in \refToFigure{notations}.
\begin{itemize}
\item $\unf(\open,\mapEnv)$ is the \emph{unfolding} of $\open$ in $\mapEnv$, that is, the corresponding object, if any. 
\item $\mapUnion{\mapEnv_1}{\mapEnv_2}$ is the \emph{union of environments}, defined if they agree on the common domain.
\item $\bisim{\Caps{\open}{\mapEnv}}{\Caps{\open'}{\mapEnv'}}$ is the \emph{equivalence of capsules}. {As will be formalized in the first part of \refToSection{soundness}, equivalent capsules denote the same sets of abstract objects. }
This equivalence is extended by congruence to expressions, in particular to calls $\MethodCall{\open_0}{\m}{\openBar}$. 
\item {$\uptobisim{\callEnv}{\mapEnv}$ is obtained by extending $\callEnv$ \emph{up to {equivalence}} in $\mapEnv$. That is, detection of already encountered calls is performed up-to {equivalence} in the current environment. }
\end{itemize}
\begin{figure}
\begin{grammatica}
\produzione{\open\in\coFJcVals}{\NewExpr{\C}{\openBar}\mid x}{open \EZ{object}}\\
\produzione{\mapEnv}{\x_1 :\open_1\ldots\ \x_n : \open_n \Space (n\geq0)}{environment}\\
\produzione{\call}{\MethodCall{\open}{\m}{\openBar}}{call}\\
\produzione{\lbcall}{[\checklabel]}{{optional checking tag}}\\
\produzione{\callEnv}{{\envItem{\call_1}{\lbcall_1}{\x_1},\ldots,\envItem{\call_n}{\lbcall_n}{\x_n}} \Space (n\ge 0)}{call trace}
\end{grammatica}
\\[2ex]

\hrule 

\[\begin{array}{l}
\NamedRule{val}{}{\opsem{\open}{\mapEnv}{\callEnv}{\open}{{\mapEnv}}}{}  
\BigSpace
\NamedRule{field}{
  \opsem{\E}{\mapEnv}{\callEnv}{\open}{\mapEnvPrime}
}{\opsem{\FieldAccess{\E}{\f}}{\mapEnv}{\callEnv}{\open_i}{\mapEnvPrime} }
{\unf(\open,\mapEnvPrime)=\NewExpr{\C}{{\open_1,\ldots,\open_n}}\\
\fields(\C)=\f_1...\f_n\\
\f=\f_i, i\in 1..n}
\\[7ex]
\NamedRule{new}{
  \opsem{\E_i}{\mapEnv}{\callEnv}{\open_i}{\mapEnv'_i}\Space  \forall i \in 1..n
}{ \opsem{\NewExpr{\C}{\E_1,\ldots,\E_n}}{\mapEnv}{\callEnv}{\NewExpr{\C}{{\open_1,\ldots,\open_n}}}{\bigmapUnion{i \in 1..n}{{\mapEnv'_i}}}}
{}
\\[8ex]
\mbox{In all the following rules:}
\begin{array}{l}
\EBar=\E_1,\ldots,\E_n\\
\openBar=\open_1\ldots\open_n\\
\call=\MethodCall{\open_0}{\m}{\openBar}\\
\mapEnvU = \bigmapUnion{i \in 0..n}{{\mapEnv'_i}} \\
\unf(\open_0,\mapEnv'_0)=\NewExpr{\C}{\_}\\
\end{array}
\\[11ex]

\NamedRule{invk-ok}{
  \begin{array}{l}
    \opsem{\E_i}{\mapEnv}{\callEnv}{\open_i}{{\mapEnv'_i}}\Space \forall i \in 0..n
    \\
    \opsem{\Subst{\Subst{\E}{\open_0}{\kwThis}}{\openBar}{\xBar}}{\mapEnvU}{\PB{\UpdateEnvCo{\callEnv}{\call}{\x}}}{\open}{\mapEnvPrime}
  \end{array}
}{ \opsem{\MethodCall{\E_0}{\m}{\EBar}}{\mapEnv}{\callEnv}{\open}{\mapEnvPrime} }
{ {\call\not\in\dom{\uptobisim{\callEnv}{\mapEnvU}}}\\
 \x\ \mbox{fresh}\\ 
\mbody(\C,\m)=\Pair{\xBar}{\E}\\
 \x \not\in\dom{\mapEnvPrime}}
\\[8ex]
\NamedRule{invk-check}{
  \begin{array}{l}
    \opsem{\E_i}{\mapEnv}{\callEnv}{\open_i}{{\mapEnv'_i}}\Space \forall i \in 0..n
    \\
    \opsem{\Subst{\Subst{\E}{\open_0}{\kwThis}}{\openBar}{\xBar}}{\mapEnvU}{\UpdateEnvCo{\callEnv}{{\call}}{\x}}{\open}{\mapEnvPrime}
    \\
    \opsem{ \Subst{\Subst{\E}{\open_0}{\kwThis}}{\openBar}{\xBar}}{\mapEnvU \sqcup\UpdateEnv{\mapEnv'}{\x}{\open}}{\UpdateEnvCk{\callEnv}{{\call}}{\x}}{\open'}{\mapEnv''}
  \end{array}
}{ \opsem{\MethodCall{\E_0}{\m}{\EBar}}{\mapEnv}{\callEnv}{\x}{\UpdateEnv{\mapEnvPrime}{\x}{\open}}}
{{\call\not\in\dom{\uptobisim{\callEnv}{\mapEnvU}}}\\
\x\ \mbox{fresh}\\ 
\mbody(\C,\m)=\Pair{\xBar}{\E}\\
\x \in\dom{\mapEnvPrime}\\
\bisim{\Caps{\x}{\UpdateEnv{\mapEnvPrime}{\x}{\open}}}{\Caps{\open'}{\mapEnv''}}
  }
\\[8ex]
\NamedRule{corec}{
\begin{array}{l}
  \opsem{\E_i}{\mapEnv}{\callEnv}{\open_i}{{\mapEnv'_i}}\Space  \forall i \in 0..n\\
  \opsem{\Subst{\Subst{\Subst{\E'}{\open_0}{\kwThis}}{\openBar}{\xBar}}{x}{\Any}}{\UpdateEnv{\mapEnvU}{\x}{\x}}{\callEnv}{\open}{{\mapEnvPrime}}
\end{array}
}{ \opsem{\MethodCall{\E_0}{\m}{\EBar}}{\mapEnv}{\callEnv}{\open}{ \UpdateEnv{\mapEnvPrime}{\x}{\x} }}
{
\uptobisim{\callEnv}{\mapEnvU}({\call})=\x\\
\combody(\C,\m)=\Pair{\xBar}{\E'}
}
\\[8ex]
\NamedRule{look-up}{
  \opsem{\E_i}{\mapEnv}{\callEnv}{\open_i}{{\mapEnv'_i}}\Space  \forall i \in 0..n
}{ \opsem{\MethodCall{\E_0}{\m}{\EBar}}{\mapEnv}{\callEnv}{\x}{ {\mapEnvU} } }
{
\uptobisim{\callEnv}{\mapEnvU}({\call})={\x^\checklabel}\\
}
\end{array}\]
\caption{$\coFJ$ operational semantics} \label{fig:new-sem}
\end{figure}
\EZComm{cut: When reading the rules, recall that they are expected to preserve the invariant that the result of evaluation $\Caps{\open}{\mapEnv}$ satisfies the capsule property, that is, $\mapEnv$ should be defined on all the variables possibly occurring in $\open$.}

Rule \refToRule{val} is needed for objects which are not $\FJ$ objects. Rule \refToRule{field} is similar to that of $\FJ$ except that the capsule $\Caps{\open}{\mapEnvPrime}$ must be unfolded to retrieve the corresponding object. Furthermore, the resulting environment is that obtained by evaluating the receiver. 
Rule \refToRule{new} is analogous to that of $\FJ$. The resulting environment is the union of those obtained by evaluating the arguments.

There are four rules for method invocation. In all of them, as in the $\FJ$ rule, the receiver and argument expressions are evaluated first to obtain the call
$\call=\MethodCall{\open}{\m}{\openBar}$. The  environment $\mapEnvU$ is the union of those obtained by these evaluations.
Then, the behavior is different depending whether such call (meaning a call equivalent to $\call$ in $\mapEnvU$) has been already encountered.

Rules \refToRule{invk-ok} and \refToRule{invk-check} handle\footnote{The two rules could be merged together, but we prefer to make explicit the difference  for sake of clarity.} a call $\call$ which is encountered the first time, as expressed by the side condition $\call\not\in{\dom{\uptobisim{\callEnv}{\mapEnvU}}}$. 
In both, the definition $\E$, where the receiver  replaces $\kwThis$ and the arguments replace the parameters, is evaluated.
Such evaluation is performed in the call trace $\callEnv$ updated to associate the call $\call$ with an unused variable $\x$ (in these two rules  ``$\x$ fresh'' means that $\x$ does not occur in the derivations of $\opsem{\E_i}{\mapEnv}{\callEnv}{\open_i}{{\mapEnv'_i}}$, for all $i\in 0..n$), and produces the capsule $\Pair{\open}{\mapEnvPrime}$.
Then there are two cases, depending on whether $\x\in\dom{\mapEnvPrime}$ holds. 

If $\x\not\in\dom{\mapEnvPrime}$, then the evaluation of the definition for $\call$ has been performed without evaluating the codefinition. That is, \EZ{the same call has not} been encountered, hence the result has been obtained by standard recursion, and no additional check is needed.

If $\x\in\dom{\mapEnvPrime}$, instead, then the evaluation of the definition for ${\call}$ has required to evaluate the codefinition. In this case, an additional check is required {(third premise)}. That is, $\Subst{\Subst{\E}{\open_0}{\kwThis}}{\openBar}{\xBar}$ is evaluated once more under the assumption that $\open$ is the result of the call. {Formally, evaluation takes  place in an environment updated to associate $\x$ with $\open$, and the variable $\x$ corresponding to the call is tagged with $\checklabel$.} The capsule result obtained in this way must be (equivalent to) that obtained by the first evaluation of the body of the method. 
{In \refToSection{examples} we discuss in detail the role of this additional check, showing an example where it is necessary.}
If the check succeeds, then the final result is the variable $\x$ in the environment updated to associate $\x$ with $\open$.
 Otherwise, rule \refToRule{invk-check} cannot be applied since the last premise does not hold. For simplicity, we assume the result of ${\call}$ to be undefined in this case; an additional rule could be added raising a runtime error in case the result is different from the expected one, as should be done in an implementation.

The remaining rules handle an already encountered call $\call$, that is, $\uptobisim{\callEnv}{\mapEnvU}(\call)$ {is defined. The behaviour is different depending on whether the corresponding variable $\x$ is tagged or not.}

{If $\x$ is not tagged,} then rule \refToRule{corec} evaluates the codefinition where the receiver object replaces $\kwThis$, the arguments replace the parameters, and, furthermore, the variable $\x$ found in the call trace replaces $\Any$. {In addition, $\mapEnvU$ is updated to
  associate  $\x$ with $\x$. In this way, the semantics keeps track of the application of rule \refToRule{corec}.} 
 
If {$\x$ is tagged, instead, then we are in a checking step for the corresponding call.} In this case, rule \refToRule{look-up} simply returns the associated variable for a call; by definition of the operational semantics, in this case such a variable is always {defined} in the environment.

{\refToFigure{notations} contains the formal definitions of the notations used in the rules.}

\begin{figure}
\begin{small}
\[\begin{array}{l}
\unf(\open,\mapEnv) =
\begin{cases}
\NewExpr{\C}{\openBar} &$ if $\open=\NewExpr{\C}{\openBar}\\
\unf(\mapEnv(\open),\mapEnv) &$ if $\open=\x\\
\end{cases}
\\ 
\Undetermined{\mapEnv}=\{\x\in\dom{\mapEnv} \mid \undef{\unf(\x,\mapEnv)} \}
\end{array}\]
\hrule
\[\begin{array}{l}
\text{For $\mapEnv_1$ and $\mapEnv_2$ such that $\mapEnv_1(\x) = \mapEnv_2(\x)$ for all $\x \in \dom{\mapEnv_1}\cap\dom{\mapEnv_2}$} \\[2ex]
(\mapUnion{\mapEnv_1}{\mapEnv_2})(\x) = 
\begin{cases}
\mapEnv_1(\x) & \x\in\dom{\mapEnv_1} \\
\mapEnv_2(\x) & \x\in\dom{\mapEnv_2} 
\end{cases} 
\end{array}\]
\hrule 
{Set $\varEquiv{\mapEnv}{}{}$ the least equivalence relation on $\Undetermined{\mapEnv}$ such that $\varEquiv{\mapEnv}{\x}{y}$ if 
 $\mapEnv( x) =  y$, {$\eqclass{\x}$ the equivalence class of $\x$,} and $\Quotient{\mapEnv}$ the quotient. 
A relation $\varRel \subseteq \Undetermined{\mapEnv_1} \times \Undetermined{\mapEnv_2}$ is a \emph{$\mapEnv_1,\mapEnv_2$-renaming}  if it induces a (partial) bijection from $\Quotient{\mapEnv_1}$, {still denoted $\varRel$,} to $\Quotient{\mapEnv_2}$.
Given $\varRel$ a $\mapEnv_1,\mapEnv_2$-renaming, the relation $\bisim[\varRel]{\Caps{\x}{\mapEnv_1}}{\Caps{\x'}{\mapEnv_2}}$ is coinductively defined by:}
\[\begin{array}{l}
\RuleNoName{}{\bisim[\varRel]{\Caps{\x}{\mapEnv}}{\Caps{\x'}{\mapEnv'}}}{
x\varRel x'
}  
\BigSpace
\RuleNoName{\bisim[\varRel]{\Caps{\open_i}{\mapEnv}}{\Caps{\open'_i}{\mapEnv'}}\quad\forall i \in 1..n}{\bisim[\varRel]{\Caps{\open}{\mapEnv}}{\Caps{\open'}{\mapEnv'}}}{\unf(\open,\mapEnv)=\NewExpr{\C}{\open_1,..,\open_n}\\ \unf(\open',\mapEnv')=\NewExpr{\C}{\open'_1,..,\open'_n}}
\end{array}\]

{A $\mapEnv_1,\mapEnv_2$-renaming $\varRel$ is \emph{strict} if,
for $\x,y\in\Undetermined{\mapEnv_1}\cap\Undetermined{\mapEnv_2}$, 
$\eqclass{\x}\varRel \eqclass{y}$ iff $\varEquiv{\mapEnv_1}{\x}{y}$ and $\varEquiv{\mapEnv_2}{\x}{y}$.}\\
We write $\bisim{\Caps{\open}{\mapEnv}}{\Caps{\open'}{\mapEnv'}}$ if $\bisim[\varRel]{\Caps{\open}{\mapEnv}}{\Caps{\open'}{\mapEnv'}}$ for some strict $\varRel$. 
\FDComm{dobbiamo dire qualcosa per il caso $\mapEnv_1 = \mapEnv_2$?}  

\hrule 

\[ \uptobisim{\callEnv}{\mapEnv}(\call')=\callEnv(\call)$ for each $\call'$ such that $\bisim{\Caps{\call'}{\mapEnv}}{\Caps{\call}{\mapEnv}} \]
\caption{$\coFJ$ auxiliary definitions} \label{fig:notations}
\end{small}
\end{figure}

Note that $\unf$, being inductively defined, can be undefined, denoted $\uparrow$, in presence of unguarded cycles among variables. 
Capsule equivalence, instead, is defined coinductively, so that, e.g., $\Caps{\x}{\x:\NewExpr{\C}{\x}}$ is equivalent to $\Caps{\x}{\x:\NewExpr{\C}{\NewExpr{\C}{\x}}}$. Capsule equivalence implicitly subsumes $\alpha$-equivalence of variables whose unfolding is defined, e.g.,  $\Caps{\x}{\x:\NewExpr{\C}{\x}}$ is equivalent to $\Caps{y}{y:\NewExpr{\C}{y}}$. Instead, $\alpha$-equivalence of undetermined variables is given by an explicit renaming, which should preserve disjointness of cycles. For instance,  $\Caps{\NewExpr{\C}{\x,y}}{(\x:y,y:\x)}$ is equivalent to $\Caps{\NewExpr{\C}{\x,\x}}{\x:\x}$, but is \emph{not} equivalent to $\Caps{\NewExpr{\C}{\x,y}}{(\x:\x,y:y)}$. Indeed, in the latter case $\x$ and $y$ can be instantiated independently. {We will prove in \refToSection{soundness} (\refToTheorem{eq-sem-caps}) that the relation $\bisim[\varRel]{}{}$, for some $\mapEnv_1,\mapEnv_2$-renaming $\varRel$, is the operational counterpart of the fact that two capsules denote the same set of abstract values. The stronger strictness condition prevents erroneous identification of objects during evaluation, e.g., $\Caps{\NewExpr{\C}{\x,y}}{(\x:\x,y:y)}$ is not equivalent to $\Caps{\NewExpr{\C}{y,\x}}{(y:y,x:x)}$.}

\smallskip
\noindent\textbf{Determinism and conservativity} 
\PB{In contrast to $\coFJ$ abstract semantics, but like $\FJ$,} $\coFJ$ operational semantics is deterministic. 

\begin{theorem}[Determinism] \label{theo:cofj-determinism}
If $\opsem{\E}{\mapEnv}{\callEnv_1}{\open_1}{\mapEnv_1}$ and $\opsem{\E}{\mapEnv}{\callEnv_2}{\open_2}{\mapEnv_2}$ hold and $\dom{\callEnv_1} = \dom{\callEnv_2}$, then 
$\Pair{\open_1}{\mapEnv_1}$ and $\Pair{\open_2}{\mapEnv_2}$ are equal up-to $\alpha$-equivalence. 
\end{theorem}
\begin{proof}
The proof is by induction on the derivation for $\opsem{\E}{\mapEnv}{\callEnv_1}{\open_1}{\mapEnv_1}$. 
The key point is that, once fixed $\E$, $\mapEnv$ and $\dom{\callEnv_1}$, there is a unique applicable rule, hence both $\opsem{\E}{\mapEnv}{\callEnv_1}{\open_1}{\mapEnv_1}$ and $\opsem{\E}{\mapEnv}{\callEnv_2}{\open_2}{\mapEnv_2}$ are derived by the same rule. 
\end{proof}

As the abstract one, the operational semantics is a conservative extension of the standard $\FJ$ semantics. {This result follows from soundness with respect to the abstract semantics in next section, however the direct proof below provides some useful insight.}

\begin{theorem}[Conservativity] \label{theo:op-conservative}
If $\validInd{\isFJ}{\BigStepFJ{\E}{\val}}$, then \PB{$\opsem{\E}{\emptyset}{\emptyset}{\EZ{\open}}{\mapEnv}$ holds \EZ{iff} $\EZ{\open} = \val$ and $\mapEnv = \emptyset$.}
\end{theorem}

For the proof, we need some auxiliary lemmas and definitions.
First, we note that $\FJ$ has the \emph{strong determinism} property: each expression has at most one finite proof tree in $\isFJ$. 
\begin{lemma}[$\FJ$ strong determinism]\label{lemma:strongdet}
If $\validInd{\isFJ}{\BigStepFJ{\E}{\val_1}}$ by a proof tree $t_1$ and $\validInd{\isFJ}{\BigStepFJ{\E}{\val_2}}$ by a proof tree $t_2$, then $t_1 = t_2$ and $\val_1=\val_2$. 
\end{lemma}
\begin{proof}
By induction on the definition of $\BigStepFJ{\E}{\val_1}$.
The key point is that each judgement is the consequence of exactly one rule.
\end{proof}

By relying on strong determinism, it is easy to see that in $\FJ$ a proof tree for an expression cannot contain another node labelled by the same expression. 
In other words, if the evaluation of $\E$ requires to evaluate $\E$ again, then the $\FJ$ semantics is undefined on $\E$, as expected.

\begin{lemma}\label{lemma:disjoint}
A proof tree in $\isFJ$ for $\BigStepFJ{\E}{\val}$ cannot contain any other  node  $\BigStepFJ{\E}{\val'}$, for any $\val'$.
\end{lemma}
\begin{proof}
By \refToLemma{strongdet}, there is a unique proof tree $t$ for the expression $\E$. 
Hence, a node $\BigStepFJ{\E}{\val'}$ in $t$  would be necessarily the root of a subtree of $t$ equal to $t$, that is, it is the root of $t$. 
\end{proof}


\begin{definition}
{Let $\validInd{\isFJ}{\BigStepFJ{\E}{\val}}$.
A call trace $\callEnv$ is \emph{disjoint} from $\BigStepFJ{\E}{\val}$  if in its proof tree\footnote{Unique thanks to \refToLemma{strongdet}.}} there are 
no instances of \refToRule{$\FJ$-invk} where $\MethodCall{\val_0}{\m}{\vBar}\in\dom{\callEnv}$.
\FDComm{si capisce chi \` e $\MethodCall{\val_0}{\m}{\vBar}$?} \EZComm{penso di s\`i, in ogni caso non saprei come migliorare}
\end{definition}

\begin{lemma}  \label{lemma:disjoint-trace}
If $\validInd{\isFJ}{\BigStepFJ{\E}{\val}}$, then, for all $\callEnv$ disjoint from $\BigStepFJ{\E}{\val}$, we have $\opsem{\E}{\emptyset}{\callEnv}{\val}{\emptyset}$.
\end{lemma}  
\begin{proof}
The proof is by induction on the definition of $\BigStepFJ{\E}{\val}$. 
\begin{description}
\item[\refToRule{\FJ-field}] 
Let $\callEnv$ be a call trace disjoint from $\BigStepFJ{\FieldAccess{\E}{\f}}{\val_i}$. 
Since $\validInd{\isFJ}{\BigStepFJ{\E}{\val}}$, with $\val = \NewExpr{\C}{\val_1,\ldots,\val_n}$, holds by hypothesis, and $\callEnv$ is, by definition,  also disjoint from $\BigStepFJ{\E}{\val}$, we get $\opsem{\E}{\emptyset}{\callEnv}{\val}{\emptyset}$ by induction hypothesis. 
Then, since $\unf(\val, \emptyset) = \val$, we get $\opsem{\FieldAccess{\E}{\f}}{\emptyset}{\callEnv}{\val_i}{\emptyset}$ by rule \refToRule{field}.
 
\item[\refToRule{$\FJ$-new}] 
Let $\callEnv$ be a call trace disjoint from $\BigStepFJ{\NewExpr{\C}{\E_1,\ldots,\E_n}}{\NewExpr{\C}{\val_1,\ldots,\val_n}}$.
For all $i \in 1..n$, since $\validInd{\isFJ}{\BigStepFJ{\E_1}{\val_i}}$ holds by hypothesis, and $\callEnv$ is, by definition, also disjoint from $\BigStepFJ{\E_i}{\val_i}$, we get $\opsem{\E_{{i}}}{\emptyset}{\callEnv}{\val_i}{\emptyset}$ by induction hypothesis. 
Then, we get $\opsem{\NewExpr{\C}{\E_1,\ldots,\E_n}}{\emptyset}{\callEnv}{\NewExpr{\C}{\val_1,\ldots,\val_n}}{\emptyset}$ by rule \refToRule{new}. 

\item[\refToRule{$\FJ$-invk}] 
Let $\callEnv$ be a call trace disjoint from $\BigStepFJ{\MethodCall{\E_0}{\m}{\E_1,\ldots,\E_n}}{\val}$.  
For all $i\in 0..n$, since $\validInd{\isFJ}{\BigStepFJ{\E_i}{\val_i}}$ holds by hypothesis, and $\callEnv$ is, by definition, also disjoint from $\BigStepFJ{\E_{{i}}}{\val_i}$, we get $\opsem{\E_i}{\emptyset}{\callEnv}{\val_i}{\emptyset}$ by induction hypothesis. 
Set $\vBar = \val_1\ldots\val_n$ and $\E' = \Subst{\Subst{\E}{\val_0}{\kwThis}}{\vBar}{\xBar}$. 
By hypothesis, $\validInd{\isFJ}{\BigStepFJ{\E'}{\val}}$ and, by definition, $\callEnv$ is also disjoint from $\BigStepFJ{\E'}{\val}$;
furthermore, by \refToLemma{disjoint}, $\E'$ cannot occur twice in the proof tree for $\BigStepFJ{\E'}{\val}$, hence $\UpdateEnvCo{\callEnv}{\MethodCall{\val_0}{\m}{\vBar}}{\x}$ is disjoint from $\BigStepFJ{\E'}{\val}$, for any fresh variable $\x$. 
Then, by induction hypothesis, we have $\opsem{\E'}{\emptyset}{\UpdateEnvCo{\callEnv}{\MethodCall{\val_0}{\m}{\vBar}}{\x}}{\val}{\emptyset}$, thus we get $\opsem{\MethodCall{\E_0}{\m}{\E_1,\ldots,\E_n}}{\emptyset}{\callEnv}{\val}{\emptyset}$ by rule \refToRule{invk-ok}. 
\end{description}
\end{proof}
 
We can now prove the conservativity result for $\coFJ$ operational semantics. 
\begin{proofOf}{\refToTheorem{op-conservative}} 
\PB{The right-to-left direction follows from \refToLemma{disjoint-trace}, since $\emptyset$ is disjoint from any expression, while the other direction follows from the right-to-left one} and \refToTheorem{cofj-determinism}. 
\end{proofOf} 

For $\coFJ$ operational semantics we can prove an additional result, characterizing derivable judgements which produce an empty environment. {The meaning is that all results obtained \emph{without using the codefinitions} are original $\FJ$ results.}

\begin{lemma}\label{lemma:op-conservative-plus}
If $\opsem{\E}{\emptyset}{\callEnv}{\open}{\emptyset}$ holds, then  \EZ{$\open$ is an $\FJ$ value $\val$, and $\validInd{\isFJ}{\BigStepFJ{\E}{\val}}$}. 
\end{lemma}

\section{Advanced examples}\label{sect:examples}

This section provides some more complex examples to better understand the operational semantics of $\coFJ$ in \refToSection{operational}
and its relationship with the abstract semantics in \refToSection{abstract}.

\smallskip
\noindent\textbf{Examples on lists} 
We first show an example motivating the additional checking step (third premise) in rule \refToRule{invk-check}. Essentially, the success of this check for some capsule result corresponds to the existence of an infinite tree in the abstract semantics, whereas the fact that this capsule result is obtained by assuming the codefinition as result of the cyclic call (second premise) corresponds to the existence of a finite tree which uses the codefinition\EZComm{questo si vede in qualche modo nella prova formale? se s\`i mettere puntatore}.  

Assume to add to our running example of lists of integers a method that returns the sum of the elements. \EZ{For infinite regular lists, that is, lists ending with a cycle, a result should be returned if the cycle has sum $0$, for instance for a list ending with infinitely many $0$s,} and no result if the cycle has sum different from $0$.
This can be achieved \EZ{as follows}.\label{sum-example}
\begin{lstlisting}
class List extends Object { ...
   int sum() {0}
}
class NonEmptyList extends List { ...
     int sum() {this.head + this.tail.sum()} corec {0}
}
\end{lstlisting}
It is easy to see that the abstract semantics of the previous section formalizes the expected behavior. 
For instance,  an infinite tree for a judgment {$\eval{\MethodCall{[2,1]^\omega}{\sumAll}{}}{\val}$ only exists for $\val=2+1+\val$,} and there are no solutions of this equation, hence there is no result. In the operational semantics, by evaluating the body assuming the codefinition as result of the cyclic call (second premise of rule \refToRule{invk-check}) the spurious result $3$ would be returned.  This is avoided by the third premise, which evaluates the method body assuming $3$ as result of the cyclic call. Since we \emph{do not} get $3$ in turn as result, evaluation is stuck, as expected. 

\EZ{Note that the stuckness situation is detected: the last side-condition of rule \refToRule{invk-check} fails, and a dynamic error (not modeled for simplicity, see the comments to the rule) is raised, likely an exception in an implementation.
On the other hand, computations which \emph{never} reach (a base case or) an already encountered call still do not terminate in this operational semantics, exactly as in the standard one, and the fact that this does not happen should be \emph{proved} by suitable techniques, see the Conclusion.}

{All the examples shown until now have a constant codefinition. We show now an example where this is not enough. Consider the method \texttt{remPos()} that removes positive elements. A first attempt at a $\coFJ$ definition is the following:}

\begin{lstlisting}
class NonEmptyList extends List { ...
  List remPos() { 
    if(this.head > 0) this.tail.remPos()
    else new NonEmptyList(this.head,this.tail.remPos())}
  corec {new EmptyList()}
\end{lstlisting}

{Is this definition correct? Actually, it provides the expected behavior on finite lists, and cyclic lists where the cycle contains only positive elements. However, when the cycle contains at least one non positive element, there is no result. For instance, consider the method call $\MethodCall{[0,1]^\omega}{\texttt{remPos}}{}$. In the abstract semantics, an infinite tree can be constructed for the judgment $\eval{\MethodCall{[0,1]^\omega}{\texttt{remPos}}{}}{\val}$ only if $\val= 0: \val$, and this clearly only holds for $\val=[0]^\omega$. However, no finite tree can be constructed for this judgment using the codefinition. Note that, in the operational semantics, without the additional check (third premise of rule \refToRule{invk-check}), we would get the spurious result $[0]$. In order to have a $\coFJ$ definition complete with respect to the expected behavior, we should provide a different codefinition for lists with \PB{infinitely many} non-positive elements.}

\begin{lstlisting}
class NonEmptyList extends List { ...
  List remPos() { 
    if(this.head > 0) this.tail.remPos()
    else new NonEmptyList(this.head,this.tail.remPos())}
  corec {if (this.allPos() then new EmptyList() else any}
\end{lstlisting}

\smallskip
\noindent\textbf{Arithmetic with rational and real numbers}\EZComm{ho tagliuzzato qualche parola per accorciare}
All real numbers in the closed interval $\{0..1\}$ can be represented
by infinite lists $[d_1,d_2,\ldots]$ of decimal digits; more precisely,
the infinite list $[d_1,d_2,\ldots]$ represents
the real number which is the limit of the series $\sum_{i=1}^{\infty}10^{-i}d_i$.

It is well-known that all rational numbers in $\{0..1\}$ correspond to either a terminating or repeating decimal,
hence they can be represented by infinite regular lists of digits, where 
terminating decimals end with either an infinite sequence of $0$ or an infinite sequence of $9$;
for instance, the terminating decimal $\frac{1}{2}$ can be represented equivalently by either $[5,0,0,\ldots]$ or $[4,9,9,\ldots]$, while
the repeating decimal $\frac{1}{3}$ is represented by $[3,3,\ldots]$.

Therefore, in $\coFJ$ all rational numbers in $\{0..1\}$ can be effectively represented with infinite precision at the level
of the operational semantics; to this aim, we can declare a class \texttt{Number} with the two fields
\texttt{digit} of type \texttt{int} and \texttt{others} of type \texttt{Number}: \texttt{digit}
contains the leftmost digit, that is, the most significant, while \texttt{others} refers to
the remaining digits, that is, the number we would obtain by a single left shift (corresponding to multiplication by $10$).
Since also non-regular values are allowed, in the abstract semantics class \texttt{Number} can be used to represent also all
irrational numbers in $\{0..1\}$.

We now show how it is possible to compute in $\coFJ$ the addition of rational numbers in  $\{0..1\}$ with infinite precision.
We first define the method \texttt{carry} which computes the carry of the addition of two 
numbers: its result is $0$ if the sum belongs to $\{0..1\}$, $1$ otherwise.

\begin{lstlisting}
class Number extends Object { // numbers in {0..1} 
  int digit; // leftmost digit
  Number others; // all other digits                  

  int carry(Number num){ // returns 0 if this+num<=1, 1 otherwise
    if (this.digit+num.digit!=9) (this.digit+num.digit)/10
    else this.others.carry(num.others)
  } corec {0}
}
\end{lstlisting}

The two numbers \texttt{this} and \texttt{num} are inspected starting from the most
significant digits: if their sum is different from $9$, then the
carry can be computed without inspecting the other digits, hence the integer division by $10$
of the sum is returned. 
Corecursion is needed when the sum of the two digits equals $9$;
in this case the carry is the same obtained from the addition of \texttt{this.others} and \texttt{num.others}.

Finally, in the codefinition the carry $0$ is returned; indeed, the codefinition is evaluated only when the sum of
the digits for all positions inspected so far is $9$ and the same patterns of digits are encountered for the second time.
This can only happen for pairs of numbers whose addition is $[9,9,\ldots]$, that is, $1$, hence the computed carry must be $0$.

\EZComm{mi sembra questo si possa tagliare non dice nulla di nuovo: In the abstract semantics, when the addition of $n_1$ and $n_2$ yields $1$, without the corules it is possible to derive $\eval{n_1\mathtt{.carry}(n_2)}{\val}$ with an infinite proof tree  for any value $\val$, but thanks to the codefinition all spurious values are filtered and the only correct value $0$ is kept;
indeed, if \texttt{0} is replaced with \texttt{any} in the codefinition, then in the operational semantics the value
returned by $n_1\mathtt{.carry}(n_2)$ is undetermined when the addition of $n_1$ and $n_2$ yields $1$, whereas in the abstract semantics any
integer is returned for the same case. }

Based on method \texttt{carry}, we can define method \texttt{add} which computes the addition of two
numbers, excluding the possible carry in case of overflow.
\begin{lstlisting}
class Number extends Object { ... // declarations as above  
  Number add(Number num){ // returns this+num 
    new Number(
      (this.digit+num.digit+this.others.carry(num.others))%10,
      this.others.add(num.others))} corec {any}
}  
\end{lstlisting}
For each position, the corresponding digits of \texttt{this} and \texttt{num} are added
to the carry computed for the other digits (\texttt{this.others.carry(num.others)}),
then the reminder of the division by $10$ \EZ{gives} the most significant digit
of the result, whereas the others are obtained by corecursively
calling the method on the remaining digits (\texttt{this.others.add(num.others)}).
Since this call is guarded by a constructor call, the codefinition
is \texttt{any}.

\EZ{Note that, in the abstract semantics, methods \texttt{carry} and \texttt{add} correctly work} also
for irrational numbers.

Method \texttt{add} above is simple, but has the drawback that
the same carries are computed more times; hence, in the worst case,
the time complexity is quadratic in the period\footnote{Indeed, the worst case scenario is when the carry
propagates over all digits because their sum is always $9$, and this can happen only if the two numbers have the same period.} of the two involved
repeating decimals. To overcome this issue, we present a more elaborate example where carries are computed only once for any position;
this is achieved by method \texttt{all\_carries} below, which returns the sequence of all carries (hence, a list of binary digits).

Method \texttt{simple\_add} corecursively adds all digits without considering carries,
while method \texttt{add}, defined on top of \texttt{simple\_add} and \texttt{all\_carries}, 
computes the final result.
This new version of \texttt{add} is not recursive and, hence. does not need a codefinition.

\begin{lstlisting}
class Number extends Object { ... // declarations as above  
  Number all_carries(Number num){ // carries for all positions
    this.simple_carries(num).complete()
  } 
  Number simple_carries(Number num){ // carries computed immediately 
    if(this.digit+num.digit!=9)
      new Number((this.digit+num.digit)/10,
        this.others.simple_carries(num.others))
    else new Number(9,this.others.simple_carries(num.others))
  } corec {any}

  Number complete(){ // computes missing carries marked with 9
    if(this.digit!=9) new Number(this.digit,this.others.complete())
    else this.fill(this.carry_lookahead()).complete()
  } corec {any}
 
  Number fill(int dig){ // fills with dig all next missing carries 
    if(this.digit!=9) this else new Number(dig,this.others.fill(dig))
  } corec {any}

  int carry_lookahead(){ // returns the next computed carry
    if(this.digit!=9) this.digit else this.others.carry_lookahead()
  } corec {0}
  
  Number simple_add(Number num){ // addition without carries 
    new Number((this.digit+num.digit)%10,
      this.others.simple_add(num.others))
  } corec {any}

  Number add(Number num){
    this.simple_add(num).simple_add(this.all_carries(num).others)
  }
}
\end{lstlisting}

\newcommand{\NatInfty}{\texttt{Nat}^\infty}
\newcommand{\Nat}{\texttt{Nat}}
\newcommand{\Infty}{\texttt{Infty}}

\smallskip
\noindent\textbf{Distances on graphs}
\EZ{The last example of this section involves graphs, which are the paradigmatic example of cyclic data structure. Our aim is to compute the \emph{distance}, that is, the minimal length of a path, between two vertexes\footnote{The example can be easily adapted to weighted paths.}. Consider a graph ($V$, \textit{adj}) where $V$ is the set of vertexes and $adj:V\rightarrow\wp(V)$ gives, for each vertex, the set of the adjacent vertexes. Each vertex has an identifier \texttt{id} assumed to be unique.  We assume a class $\NatInfty$, with subclasses $\Nat$ with an integer field, and $\Infty$ with no fields, for naturals and $\infty$ (distance between unconnected nodes), respectively.
Such classes offer methods \texttt{succ()} for the successor, and \texttt{min($\NatInfty$ n)} for the minimum, with the expected behaviour (e.g., \texttt{succ} in class $\NatInfty$ returns $\infty$).}

\begin{lstlisting}
class Vertex extends Object {
  Id id; AdjList adjVerts;
  $\NatInfty$dist(Id id) {
    this.id==id?new Nat(0):this.adjVerts.dist(id).succ()}
  corec {new $\Infty$()}
}

class AdjList extends Object { }
class EAdjList extends AdjList {
  $\NatInfty$dist(Id id) { new $\Infty$() }
}
class NEAdjList extends AdjList {
  Vertex vert; AdjList adjVerts;
  $\NatInfty$dist(Id id) {this.vert.dist(id).min(this.adjVerts.dist(id))}
}
\end{lstlisting}

\EZ{Clearly, if the destination \texttt{id} and the source node coincide, then the distance is 0. Otherwise, the distance is obtained by incrementing by one the minimal distance from an adjacent to \texttt{id}, computed by method \texttt{dist()} of \texttt{AdjList} called on the adjacency list. The codefinition of method \texttt{dist()} of class \texttt{Vertex} is needed since, in presence of a cycle, $\infty$ is returned and non-termination is avoided. The same approach can be adopted for visiting a graph: instead of keeping trace of already encountered nodes, cycles are implicitly handled by the loop detection mechanism of $\coFJ$.}

\section{Soundness}\label{sect:soundness}
Soundness of the operational semantics with respect to the abstract one means, roughly, that a value derived using the rules in \refToFigure{new-sem} can also be derived by those in \refToFigure{syntax}. 
However, this statement needs to be refined, since values in the two semantics are different: \PB{possibly infinite objects in the abstract semantics, and capsules in the operational semantics.}

We define a relation from capsules to abstract objects, formally express soundness through this relation, and introduce an intermediate semantics to carry out the proof in two steps.

\smallskip
\noindent\textbf{From capsules to infinite objects} 
Intuitively, given a capsule $\Caps{\open}{\mapEnv}$, we get an abstract value by instantiating variables in $\open$ with abstract values, in a way consistent with $\mapEnv$. 
To make this formal, we need some preliminary definitions.

A \emph{substitution}  {$\subst$} is a function from variables to abstract values. We denote by $\ApplySubst{\E}{\subst}$ the {abstract} expression obtained by applying $\subst$ to $\E$.
In particular, if $\E$ is an open value $\open$, then $\ApplySubst{\open}{\subst}$ is an abstract value. 
\EZ{Given an environment $\mapEnv$ and a substitution $\subst$, the substitution  $\AppMap{\mapEnv}{\subst}$ is defined by}: 
\[ \AppMap{\mapEnv}{\subst}(\x) = \begin{cases}
\ApplySubst{\mapEnv(\x)}{\subst} & \x \in \dom{\mapEnv} \\
\subst(x) & \x \notin \dom{\mapEnv}
\end{cases} \]
Then, a \emph{solution} of $\mapEnv$ is a substitution $\subst$ such that $\AppMap{\mapEnv}{\subst} = \subst$. Let $\solutions(\mapEnv)$ be the set of solutions {of} $\mapEnv$.  
{Finally, if $\Caps{\E}{\mapEnv}$ is a capsule, we define the set of abstract expressions it denotes as 
$\semCaps{\E}{\mapEnv} = \{ \ApplySubst{\E}{\subst} \mid \subst \in \solutions(\mapEnv) \}$.
Note that $\semCaps{\open}{\mapEnv} \subseteq \coFJaVals$, for any capsule $\Caps{\open}{\mapEnv}$. 
We now show {an operational characterization of} the semantic equality.

\begin{theorem} \label{theo:eq-sem-caps}
$\semCaps{\open_1}{\mapEnv_1} {=}\semCaps{\open_2}{\mapEnv_2}$ iff $\bisim[\varRel]{\Caps{\open_1}{\mapEnv_1}}{\Caps{\open_2}{\mapEnv_2}}$, for some \mbox{$\mapEnv_1,\mapEnv_2$-renaming $\varRel$.}
\end{theorem}

To prove this result we need some auxiliary definitions and lemmas. 
The \emph{tree expansion} of a capsule $\Caps{\open}{\mapEnv}$ is the possibly  infinite open value coinductively defined as follows: 
\[\expCaps{\open}{\mapEnv} = \begin{cases}
\x  &  \open = \x \text{ and } \undef{\unf(\x, \mapEnv)} \\
\NewExpr{\C}{\expCaps{\open_1}{\mapEnv}, \ldots, \expCaps{\open_n}{\mapEnv}} & \unf(\open, \mapEnv) = \NewExpr{\C}{\open_1,\ldots,\open_n} 
\end{cases}\]
The next proposition shows relations between solutions and tree expansion of a capsule. 

\begin{proposition} \label{prop:solutions}
Let $\Caps{\open}{\mapEnv}$ be a capsule and $\subst \in \solutions(\mapEnv)$, then 
\begin{enumerate}
\item\label{prop:solutions:1} if $\undef{\unf(\open,\mapEnv)}$ then $\open = \x$ and $\varEquiv{\mapEnv}{\x}{\x}$
\item\label{prop:solutions:2} $\FV{\expCaps{\open}{\mapEnv}} \subseteq \{\x \in \dom{\mapEnv} \mid \varEquiv{\mapEnv}{\x}{\x} \}$
\item\label{prop:solutions:3} if $\varEquiv{\mapEnv}{x}{y}$ then $\subst(x) = \subst(y)$
\item\label{prop:solutions:4} if $\unf(\open, \mapEnv) = \NewExpr{\C}{\open_1,\ldots,\open_n}$ then $\ApplySubst{\open}{\subst} = \NewExpr{\C}{\ApplySubst{\open_1}{\subst},\ldots,\ApplySubst{\open_n}{\subst}}$
\item\label{prop:solutions:5} $\ApplySubst{\open}{\subst} = \ApplySubst{\expCaps{\open}{\mapEnv}}{\subst}$ 
\end{enumerate}
\end{proposition} 

Given a relation $\varRel$ on variables, we will denote by $\varRel\op$ the opposite relation and by $=_\varRel$ the equality of possibly infinite open values up-to $\varRel$, coinductively defined by the following rules: 
\[
\RuleNoName{}{x =_\varRel y}{x\varRel y} 
\BigSpace
\RuleNoName{
  t_i =_\varRel s_i\quad \forall i \in 1..n 
}{ \NewExpr{\C}{t_1, \ldots, t_n} =_\varRel \NewExpr{\C}{s_1, \ldots, s_n} }
{}
\]
It is easy to check that 
\begin{itemize}
\item $\varRel$ is a $\mapEnv_1,\mapEnv_2$-renaming iff $\varRel\op$ is a $\mapEnv_2,\mapEnv_1$-renaming,
\item $\bisim[\varRel]{\Caps{\open_1}{\mapEnv_1}}{\Caps{\open_2}{\mapEnv_2}}$ iff $\bisim[\varRel\op]{\Caps{\open_2}{\mapEnv_2}}{\Caps{\open_1}{\mapEnv_1}}$, 
\item $t_1 =_\varRel t_2$ iff $t_2 =_{\varRel\op} t_1$.
\end{itemize}

We have the following lemmas:
\begin{lemma} \label{lemma:bisim-tree}
$\bisim[\varRel]{\Caps{\open_1}{\mapEnv_1}}{\Caps{\open_2}{\mapEnv_2}}$ iff $\expCaps{\open_1}{\mapEnv_1} {=_\varRel} \expCaps{\open_2}{\mapEnv_2}$, \mbox{for each $\mapEnv_1,\mapEnv_2$-renaming $\varRel$}. 
\end{lemma}
\begin{proof}
The proof is immediate by coinduction in both directions.
\end{proof}

\begin{lemma} \label{lemma:tree-sem-sound}
If $\expCaps{\open_1}{\mapEnv_1} {=_\varRel} \expCaps{\open_2}{\mapEnv_2}$, where $\varRel$ is a $\mapEnv_1,\mapEnv_2$-renaming, then $\semCaps{\open_1}{\mapEnv_1} = \semCaps{\open_2}{\mapEnv_2}$.
\end{lemma}

\begin{proposition}\label{prop:sem-eq}
If $\semCaps{\open_1}{\mapEnv_1} = \semCaps{\open_2}{\mapEnv_2}$ then 
\begin{enumerate}
\item\label{prop:sem-eq:1} if $\undef{\unf(\open_1,\mapEnv_1)}$ then $\undef{\unf(\open_2,\mapEnv_2)}$, 
\item\label{prop:sem-eq:2} if $\unf(\open_1,\mapEnv_1) = \NewExpr{\C}{\open_{1,1},\ldots,\open_{1,n}}$ then $\unf(\open_2,\mapEnv_2) = \NewExpr{\C}{\open_{2,1},\ldots,\open_{2,n}}$ and, for all $i \in 1..n$, $\semCaps{\open_{1,i}}{\mapEnv_1} = \semCaps{\open_{2,i}}{\mapEnv_2}$. 
\end{enumerate}
\end{proposition}

\begin{lemma} \label{lemma:tree-sem-complete}
If $\semCaps{\open_1}{\mapEnv_1} {=} \semCaps{\open_2}{\mapEnv_2}$ then $\expCaps{\open_1}{\mapEnv_1} =_\varRel \expCaps{\open_2}{\mapEnv_2}$, for some $\mapEnv_1,\mapEnv_2$-renaming $\varRel$. 
\end{lemma}

\begin{proofOf}{\refToTheorem{eq-sem-caps}} 
The right-to-left direction follows from \refToLemma{bisim-tree} and \refToLemma{tree-sem-sound}, 
while the other direction follows from \refToLemma{tree-sem-complete} and \refToLemma{bisim-tree}. 
\end{proofOf}

Since by definition $\bisim{}{}$ is equal to $\bisim[\varRel]{}{}$ for some $\varRel$, applying \refToLemma{bisim-tree} and \refToLemma{tree-sem-sound} we get that if $\bisim{\Caps{\open_1}{\mapEnv_1}}{\Caps{\open_2}{\mapEnv_2}}$ then $\semCaps{\open_1}{\mapEnv_1} = \semCaps{\open_2}{\mapEnv_2}$. 
Actually we can prove a stronger result: 

\begin{lemma} \label{lemma:strict-bisim}
If $\bisim[\varRel]{\Caps{\open_1}{\mapEnv_1}}{\Caps{\open_2}{\mapEnv_2}}$ for some strict $\mapEnv_1,\mapEnv_2$-renaming $\varRel$, then, 
for each solution $\subst \in \solutions(\mapEnv_1\cap\mapEnv_2)$, there are $\subst_1 \in \solutions(\mapEnv_1)$ and $\subst_2\in\solutions(\mapEnv_2)$ such that
$\ApplySubst{\open_1}{\subst_1} = \ApplySubst{\open_2}{\subst_2}$ and, 
for all $\x \in \dom{\mapEnv_1\cap\mapEnv_2}$, $\subst_1(\x) = \subst(\x) = \subst_2(\x)$. 
\end{lemma}
}

\smallskip
\noindent\textbf{Soundness statement}
We can now formally state the soundness result:
\begin{theorem} \label{theo:soundness-main}
If $\opsem{\E}{\emptyset}{\emptyset}{\open}{\mapEnv}$, then, for all $\val \in \semCaps{\open}{\mapEnv}$, $\valid{\isFJ}{\coisFJ}{\eval{\E}{\val}}$.
\end{theorem}

This main result is about the evaluation of source expressions, \EZ{hence} both the environment and the call trace are empty. 
To carry out the proof we need to generalize the statement.\EZComm{cambiato un po' l'ordine}
\begin{theorem}[\PB{Soundness}]\label{theo:soundness}
If  $\opsem{\E}{\mapEnv}{\emptyset}{\open}{\mapEnvPrime}$, then, for all $\subst\in \solutions(\mapEnvPrime)$, $\valid{\isFJ}{\coisFJ}{\eval{\ApplySubst{\E}{\subst}}{\ApplySubst{\open}{\subst}}}$. 
\end{theorem}
\EZ{To show that this is actually a generalization, set $\mapEnv_1\leq\mapEnv_2$ if $\dom{\mapEnv_1}\subseteq\dom{\mapEnv_2}$, and, for all $\x\in\dom{\mapEnv_1}$, $\mapEnv_1(\x)=\mapEnv_2(\x)$. We use the following lemmas.}
\begin{lemma} \label{lemma:env-leq-sol}
If $\mapEnv_1\leq\mapEnv_2$, then $\solutions(\mapEnv_2) \subseteq \solutions(\mapEnv_1)$. 
\end{lemma} 
\EZComm{cut: \begin{proof} 
Let $\subst \in \solutions(\mapEnv_2)$. 
Since $\mapEnv_1\leq\mapEnv_2$, we have that, for all $\x \in \dom{\mapEnv_1}$, $\mapEnv_1(\x)=\mapEnv_2(\x)$. 
Hence, $\AppMap{\mapEnv_1}{\subst}(\x) = \ApplySubst{\mapEnv_1(\x)}{\subst}=\ApplySubst{\mapEnv_2(\x)}{\subst} = \AppMap{\mapEnv_2}{\subst}(\x) = \subst(\x)$.
Thus we get $\subst \in \solutions(\mapEnv_2)$. 
\end{proof}}

\begin{lemma} \label{lemma:env_contained}
If $\opsem{\E}{\mapEnv}{\callEnv}{\open}{\mapEnvPrime}$, then $\mapEnv\leq\mapEnvPrime$. 
\end{lemma}
In the statement of \refToTheorem{soundness}, thanks to \refToLemma{env_contained}, we know that $\mapEnv \leq \mapEnvPrime$, hence, by \refToLemma{env-leq-sol}, 
$\subst \in \solutions(\mapEnv)$, thus $\ApplySubst{\E}{\subst} \in \semCaps{\E}{\mapEnv}$. 
\refToTheorem{soundness} implies \refToTheorem{soundness-main}, since, when $\mapEnv=\emptyset$, $\E$ is closed, hence $\ApplySubst{\E}{\subst} = \E$, and all elements in $\semCaps{\open}{\mapEnvPrime}$ have shape $\ApplySubst{\open}{\subst}$ with $\subst \in \solutions(\mapEnvPrime)$. 

\smallskip
\noindent\textbf{Proof through intermediate semantics} 
In order to prove \refToTheorem{soundness}, we introduce a new semantics called \emph{intermediate}\EZ{, defined in  \refToFigure{coFJIntSem}. Values are those of the abstract semantics, hence calls are of shape $\MethodCall{\val}{\m}{\vBar}$ (\emph{abstract} calls). 
The judgment has shape $\intsem{\E}{\resEnv}{\callSet}{\val}{\callSetRes}$, with $\callSet,\callSetRes$ sets of abstract calls, $\resEnv$ map from abstract calls to values. Comparing with $\opsem{\E}{\mapEnv}{\callEnv}{\open}{\mapEnvPrime}$ in the operational semantics, no variables are introduced for calls;
$\resEnv$ and $\callSet$ play the role of the $\checklabel$ and non $\checklabel$ part of  $\callEnv$, respectively, keeping trace of already encountered calls. Moreover, $\resEnv$ directly associates to a call its value to be used in the checking step, which in $\mapEnv$ is associated to the corresponding variable. Finally, $\callSetRes$ plays the role of $\mapEnvPrime$, tracing the calls for which the codefinition has been evaluated, hence the checking step will be needed. This correspondence is made precise below.
The rules are analogous to those of \refToFigure{new-sem}, with the difference that, for an already encountered call $\call\in\callSet$, either rule \refToRule{IN-invk-ok} or rule \refToRule{IN-corec} can be applied. In other words, evaluation of the codefinition is not necessarily triggered when the \emph{first} cycle is detected. This non-determinism makes the relation with the abstract semantics simpler. }
\EZComm{congettura: la semantica intermedia dovrebbe essere anche completa rispetto a quella astratta limitata ad alberi regolari, vedi conclusioni}

\begin{figure}
\begin{grammatica}
\produzioneCo{\val\in\coFJaVals}{\NewExpr{\C}{\vBar}}{possibly infinite object}\\
\produzione{\call}{\MethodCall{\val}{\m}{\vBar}}{abstract call}\\
\produzione{\callSet}{\call_1\ \ldots\ \call_n \Space (n\geq0)}{set of \EZ{abstract} calls}\\
\produzione{\resEnv}{\call_1:\val_1 \ldots \call_n:\val_n \Space (n\geq0)}{}\\
\end{grammatica}
\\[2ex]

\hrule 

\[\begin{array}{l}
\NamedRule{IN-val}{}{\intsem{\val}{\resEnv}{\callSet}{\val}{\emptyset}}{}  
\BigSpace
\NamedRule{IN-field}{
  \intsem{\E}{\resEnv}{\callSet}{\val}{\callSetRes}
}{\intsem{\FieldAccess{\E}{\f}}{\resEnv}{\callSet}{\val_i}{\callSetRes} }
{\val=\NewExpr{\C}{\val_1,\ldots,\val_n}\\
\fields(\C)=\f_1...\f_n\\
\f=\f_i, i\in 1..n}
\\[6ex]
\NamedRule{IN-new}{
  \intsem{\E_i}{\resEnv}{\callSet}{\val_i}{\callSet'_i}\Space  \forall i \in 1..n
}{ \intsem{\NewExpr{\C}{\E_1,\ldots,\E_n}}{\resEnv}{\callSet}{\NewExpr{\C}{\val_1,\ldots,\val_n}}{\bigcup_{i \in 1..n}{\callSetRes_i}}}
{}
\\[5ex]
\mbox{In all the following rules:}
\begin{array}{l}
\EBar=\E_1,\ldots,\E_n\\
\vBar=\val_1\ldots\val_n\\
\call=\MethodCall{\val_0}{\m}{\vBar}\\
\val_0=\NewExpr{\C}{\_}\\
\end{array}
\\[7ex]

\NamedRule{IN-invk-ok}{
  \begin{array}{l}
    \intsem{\E_i}{\resEnv}{\callSet}{\val_i}{\callSetRes_i}\Space \forall i \in 0..n
    \\
    \intsem{\Subst{\Subst{\E}{\val_0}{\kwThis}}{\vBar}{\xBar}}{\resEnv}{\callSet\cup\{\call\}}{\val}{\callSetRes}
  \end{array}
}{ \intsem{\MethodCall{\E_0}{\m}{\EBar}}{\resEnv}{\callSet}{\val}{\bigcup_{i \in 0..n} \callSetRes_i {\cup}\callSetRes} }
{  \EZ{\call \not\in\callSetRes\ \mbox{or}\ \call\in\callSet}\\
\mbody(\C,\m)=\Pair{\xBar}{\E}}
\\[7ex]

\NamedRule{IN-invk-check}{
  \begin{array}{l}
    \intsem{\E_i}{\resEnv}{\callSet}{\val_i}{{\callSetRes_i}}\Space \forall i \in 0..n
    \\
    \intsem{\Subst{\Subst{\E}{\val_0}{\kwThis}}{\vBar}{\xBar}}{\resEnv}{\callSet\cup\{\call\}}{\val}{\callSetRes}
    \\
    \intsem{ \Subst{\Subst{\E}{\val_0}{\kwThis}}{\vBar}{\xBar}}{\UpdateEnv{\resEnv}{\call}{\val}}{\callSet}{\val}{\callSetResCheck}
  \end{array}
}{ 
\intsem{\MethodCall{\E_0}{\m}{\EBar}}{\resEnv}{\callSet}{\val}{\bigcup_{i \in 0..n} \callSetRes_i\cup(\callSetRes\setminus\{\call\})}}
{ \EZ{\call\not \in\callSet}\\
\mbody(\C,\m)=\Pair{\xBar}{\E}\\
 \call \in\callSetRes
  }
 
\\[7ex]
\NamedRule{IN-corec}{
\begin{array}{l}
  \intsem{\E_i}{\resEnv}{\callSet}{\val_i}{\callSetRes_i}\Space  \forall i \in 0..n\\
  \intsem{\Subst{\Subst{\Subst{\E'}{\val_0}{\kwThis}}{\vBar}{\xBar}}{u}{\Any}}{\resEnv}{\callSet}{\val}{\callSetRes}
\end{array}
}{ \intsem{\MethodCall{\E_0}{\m}{\EBar}}{\resEnv}{\callSet}{\val}{ \bigcup_{i \in 0..n} \callSetRes_i\cup\callSetRes\cup\{\call\}}}
{
\call\in\callSet\\
\combody(\C,\m)=\Pair{\xBar}{\E'}\\
\call\not\in\dom{\resEnv}
}
\\[7ex]
\NamedRule{IN-look-up}{
  \intsem{\E_i}{\resEnv}{\callSet}{\val_i}{\callSetRes_i}\Space  \forall i \in 0..n
}{ \intsem{\MethodCall{\E_0}{\m}{\EBar}}{\resEnv}{\callSet}{\val}{\bigcup_{i \in 0..n} \callSetRes_i} }
{
\resEnv(\call)=\val
}
\end{array}\]
\caption{$\coFJ$ intermediate semantics} \label{fig:coFJIntSem}
\end{figure}

By relying on the intermediate semantics, we can prove \refToTheorem{soundness} by two steps:
\begin{enumerate}
\item The operational semantics is sound w.r.t.\ the intermediate semantics (\refToTheorem{concrToInt}).
\item The intermediate semantics is sound w.r.t.\ the abstract semantics (\refToTheorem{intToAbs}).
\end{enumerate}
At the beginning of {\refToSection{operational}}, we mentioned two issues for an operational semantics: representing infinite objects in a finite way, and replacing infinite (regular) proof trees by finite proof trees.  This proof technique nicely shows that the two issues are orthogonal: notably, detection of cyclic calls is independent from the format of values.

{To express the soundness of the operational semantics w.r.t.\ the intermediate one, we need to formally relate the two judgments. \EZ{First of all, a call trace $\callEnv$ is the disjoint union of two maps $\callEnvCk$ and $\callEnvCo$ into tagged and non-tagged variables, respectively. Then, given an environment $\mapEnv$, we define the following sets of (operational) calls:
\begin{itemize}
\item $\callSetTr{\callEnv}{}=\dom{\callEnvCo}$
\item $\callSetOp{\callEnv}{\mapEnv}{} = \dom{\mapEnv \circ \callEnvCo}$, where $\circ$ is the composition of partial functions
\item $\callSetDif{\callEnv}{\mapEnv}{\mapEnv'}{} = \callSetOp{\callEnv}{\mapEnv'}{} \setminus \callSetOp{\callEnv}{\mapEnv}{}$
\end{itemize}
For $\callSet$ set of calls and $\subst$ substitution, we abbreviate by $\callSet_\subst$ the set of abstract calls $\ApplySubst{\callSet}{\subst}$.
Note that $\callSetOp{\callEnv}{\mapEnv}{\subst} \subseteq \callSetTr{\callEnv}{\subst}$ and, if $\mapEnv_1\leq\mapEnv_2$, then $\callSetOp{\callEnv}{\mapEnv_1}{\subst} \subseteq \callSetOp{\callEnv}{\mapEnv_2}{\subst}$. Finally, $\resEnvOp{\callEnv}{\subst}(\ApplySubst{\call}{\subst}) = \val$ iff $\val = \subst(\callEnvCk(\call))$.}


Then, the soundness result can be stated as follows:

\begin{theorem}[Soundness operational w.r.t.\ intermediate] \label{theo:concrToInt}
If $\opsem{\E}{\mapEnv}{\callEnv}{\open}{\mapEnvPrime}$ then, 
for all $\subst\in \solutions(\mapEnv')$, 
there exists $\callSet$ such that 
$\callSetDif{\callEnv}{\mapEnv}{\mapEnv'}{\subst} \subseteq \callSet \subseteq \callSetOp{\callEnv}{\mapEnv'}{\subst}$ and,
$\intsem{\ApplySubst{\E}{\subst}}{\resEnvOp{\callEnv}{\subst}}{\callSetTr{\callEnv}{\subst}}{\ApplySubst{\open}{\subst}}{\callSet}$.
\end{theorem}

\EZ{In particular, the bounds on $\callSet$ ensure that it is empty when $\callEnv=\emptyset$. Hence, if $\opsem{\E}{\mapEnv}{\emptyset}{\open}{\mapEnvPrime}$ (hypothesis of \refToTheorem{soundness}), then $\intsem{\ApplySubst{\E}{\subst}}{\emptyset}{\emptyset}{\ApplySubst{\open}{\subst}}{\emptyset}$, that is, the hypothesis of \refToTheorem{intToAbs} below holds.} 

The proof of the theorem uses the following corollary of \refToLemma{strict-bisim}. 
\begin{corollary} \label{cor:strict-bisim} 
If $\bisim{\Caps{\open_1}{\mapEnv_1}}{\Caps{\open_2}{\mapEnv_2}}$, $\subst_1{\in}\solutions(\mapEnv_1)$, $\mapEnv_1\leq \mapEnv_2$, then there is $\subst_2\in\solutions(\mapEnv_2)$ such that $\ApplySubst{\open_1}{\subst_1} = \ApplySubst{\open_2}{\subst_2}$ and, for all $x \in \dom{\mapEnv_1}$, $\subst_1(x) = \subst_2(x)$. Moreover, if $\mapEnv_1 = \mapEnv_2$, then $\ApplySubst{\open_1}{\subst_1} = \ApplySubst{\open_2}{\subst_1}$.
\end{corollary}
\EZComm{cut: \begin{proof}
Immediate from \refToLemma{strict-bisim}, since, if $\mapEnv_1\leq\mapEnv_2$ then $\mapEnv_1\cap\mapEnv_2 = \mapEnv_1$, and, 
if $\mapEnv_1 = \mapEnv_2$ then $\mapEnv_1\cap\mapEnv_2 = \mapEnv_1$. 
\end{proof}}

We now state the second step of the proof: the soundness result of the intermediate semantics with respect to the abstract semantics.
\begin{theorem}[Soundness intermediate w.r.t.\ abstract]\label{theo:intToAbs}
If $\intsem{\E}{\emptyset}{\emptyset}{\val}{\emptyset}$, then \mbox{\EZ{$\valid{\isFJ}{\coisFJ}{\eval{\E}{\val}}$}}.
\end{theorem}
The proof uses the bounded coinduction principle (\refToTheorem{bcoind}), and requires some lemmas. 
} Recall that $\evalExtended{\E}{\val}$ means that the judgment $\eval{\E}{\val}$ has a finite proof tree in the (standard) inference system consisting of $\FJ$ rules and $\coFJ$ corules. }

\begin{lemma}\label{lemma:boundedness}
If $\intsem{\E}{\emptyset}{\callSet}{\val}{\callSetRes}$ then $\evalExtended{\E}{\val}$ holds.\EZComm{prova in appendice}
\end{lemma}

\begin{lemma}\label{lemma:callSet}
\EZComm{sembrava inutile: Set $\call{=}\MethodCall{\val_0}{\m}{\val_1\ldots\val_n}$.} If $\intsem{\E}{\resEnv}{\callSet\cup\{\call\}}{\val}{\callSetRes}$ holds, and $\EZ{\call \not\in\callSetRes\EZComm{cut:\ \mbox{or}\ \call\in\callSet}}$, then $\intsem{\E}{\resEnv}{\callSet}{\val}{\callSetRes}$.\EZComm{prova in appendice}
\end{lemma}

\begin{lemma}\label{lemma:rho}\EZComm{prova in appendice}
\EZComm{sembrava inutile: Set $\call=\MethodCall{\val_0}{\m}{\val_1\ldots\val_n}$.} If $\intsem{\E}{\UpdateEnv{\resEnv}{\call}{\val'}}{\callSet}{\val}{\callSetRes}$ and $\intsem{\call}{\resEnv}{\callSet}{\val'}{\emptyset}$, then $\intsem{\E}{\resEnv}{\callSet}{\val}{\callSetRes}$.
\end{lemma}

\noindent We can now prove \refToTheorem{intToAbs}.

\begin{proofOf}{\refToTheorem{intToAbs}}

We take as specification the set $A=\{\Pair{\E}{\val} \mid \intsem{\E}{\emptyset}{\emptyset}{\val}{\emptyset}\}$, and we use bounded coinduction (\refToTheorem{bcoind}).
We have to prove the following:
\begin{description}
\item[Boundedness] For all $\Pair{\E}{\val}\in A$, $\evalExtended{\E}{\val}$ holds.
\item[Consistency] For all $\Pair{\E}{\val}\in A$, there exist a rule in the abstract semantics having $\eval{\E}{\val}$ as consequence, and such that all its premises are elements of $A$.
\end{description}

\noindent Boundedness follows immediately from \refToLemma{boundedness}. 
We now prove consistency.

Consider a pair $\Pair{\E}{\val} \in A$, hence we know that $\intsem{\E}{\emptyset}{\emptyset}{\val}{\emptyset}$ is derivable. 
We proceed by case analysis on the last applied rule in the derivation of this judgement.
\begin{description}

\item [\refToRule{IN-val}] We know that \EZ{$\E=\val=\NewExpr{\C}{\val_1,\ldots,\val_n}$}. We choose as candidate rule \refToRule{abs-new}. We have to show that, for all $i\in1..n$,  $\Pair{\val_i}{\val_i}\in A$, that is, $\intsem{\val_i}{\emptyset}{\emptyset}{\val_i}{\emptyset}$ holds We can get the thesis thanks to rule \EZ{\refToRule{IN-val}}.

\item [\refToRule{IN-field}]  We know that $\E = \FieldAccess{\E'}{\f}$ and $\intsem{\E'}{\emptyset}{\emptyset}{\NewExpr{\C}{\val_1\ldots\val_n}}{\emptyset}$. We choose as candidate rule \refToRule{abs-field}, with conclusion $\eval{\FieldAccess{\E'}{\f}}{\val_i}$. We have to show that $\Pair{\E'}{\NewExpr{\C}{\val_1\ldots\val_v}}\in A$, that is, $\intsem{\E'}{\emptyset}{\emptyset}{\NewExpr{\C}{\val_1\ldots\val_v}}{\emptyset}$ holds, but this is true by hypothesis.

\item [\refToRule{IN-new}] We know that  $\intsem{\E_i}{\emptyset}{\emptyset}{\val_i}{\emptyset}$ holds for all $i\in1..n$. We choose as candidate rule \refToRule{abs-new}. We have to show that, for all $i\in1..n$,  $\Pair{\E_i}{\val_i}\in A$, that is, $\intsem{\E_i}{\emptyset}{\emptyset}{\val_i}{\emptyset}$ holds,  but this is true by hypothesis.

\item [\refToRule{IN-invk-ok}] We know that $\E=\MethodCall{\E_0}{\m}{\EBar}$, $\intsem{\E_i}{\emptyset}{\emptyset}{\val_i}{\emptyset}$ holds for all $i\in 0..n$,
\EZ{$\call=\MethodCall{\val_0}{\m}{\vBar}$, $\mbody(\C,\m)=\Pair{\xBar}{\E'}$, and $\intsem{\Subst{\Subst{\E'}{\val_0}{\kwThis}}{\vBar}{\xBar}}{\emptyset}{\{\call\}}{\val}{\emptyset}$ holds.} We choose as candidate rule \refToRule{abs-invk}. We have to show that, for all $i\in0..n$, $\Pair{\E_i}{\val_i} \in A$, and $\Pair{\Subst{\Subst{\E'}{\val_0}{\kwThis}}{\vBar}{\xBar}}{\val}\in A$. That is, that the following judgments hold: $\intsem{\E_i}{\emptyset}{\emptyset}{\val_i}{\emptyset}$ for all $i\in0..n$, and $\intsem{\Subst{\Subst{\E'}{\val_0}{\kwThis}}{\vBar}{\xBar}}{\emptyset}{\emptyset}{\val}{\emptyset}$. The judgments in the first set hold by hypothesis. The last judgment holds thanks to \refToLemma{callSet}\EZ{, where $\callSet'=\emptyset$}.

\item [\refToRule{IN-invk-check}] We know that $\E=\MethodCall{\E_0}{\m}{\EBar}$, $\intsem{\E_i}{\emptyset}{\emptyset}{\val_i}{\emptyset}$ holds for all $i\in 0..n$,
\EZ{$\call=\MethodCall{\val_0}{\m}{\vBar}$, $\mbody(\C,\m)=\Pair{\xBar}{\E'}$,} and $\intsem{\Subst{\Subst{\E'}{\val_0}{\kwThis}}{\vBar}{\xBar}}{\UpdateEnv{}{\call}{\val}}{\emptyset}{\val}{\emptyset}$ holds. 
We choose as candidate rule \refToRule{abs-invk}. We have to show that for all $i\in0..n$, $\Pair{\E_i}{\val_i} \in A$, and $\Pair{\Subst{\Subst{\E'}{\val_0}{\kwThis}}{\vBar}{\xBar}}{\val}\in A$. That is, that the following judgments hold: $\intsem{\E_i}{\emptyset}{\emptyset}{\val_i}{\emptyset}$ for all $i\in0..n$, and $\intsem{\Subst{\Subst{\E'}{\val_0}{\kwThis}}{\vBar}{\xBar}}{\emptyset}{\emptyset}{\val}{\emptyset}$. The judgments in the first set hold by hypothesis. The last judgment holds thanks to \refToLemma{rho}, 
since from the hypothesis we easily get $\intsem{\call}{\emptyset}{\emptyset}{\val}{\emptyset}$.

\item [\refToRule{IN-corec}] Empty case since to apply the rule it should be $\callSet\neq\emptyset$.

\item [\refToRule{IN-look-up}] Empty case since to apply the rule it should be \EZ{$\resEnv\neq\emptyset$}.

\end{description}
\end{proofOf}

\section{Related work}\label{sect:related}
As already mentioned, the idea of regular corecursion (keeping \PB{track} of pending method calls, so to detect cyclic calls), originates from co-SLD resolution \cite{Simon06,SimonEtAl07,AnconaDovier15}.
Making regular corecursion \emph{flexible} means that the programmer can specify the behaviour in case a cycle is detected. Language constructs to achieve such flexibility have been proposed in the logic  \cite{Ancona13,AnconaDZ17}, functional \cite{Jeannin17}, and object-oriented \cite{AnconaZ12,AnconaZ13} paradigm. 

\smallskip
\noindent\textbf{Logic paradigm} The above mentioned \emph{co-SLD resolution} \cite{Simon06,SimonEtAl07,AnconaDovier15} is a sound resolution procedure based on cycle detection. That is, the interpreter keeps track of resolved atoms and an atom selected from  the current goal can be resolved if it unifies with an atom that has been already resolved. In this way it is possible to define coinductive predicates. Correspondingly, models are subsets of the \emph{complete Herbrand basis}, that is, the set of ground atoms built on arbitrary (finite or infinite) terms, and the declarative semantics is the greatest fixed point of the monotone function associated with a program.
Structural resolution~\cite{KPS12-2,KJS17} (a.k.a. S-resolution)
is a proposed generalization for cases when formulas computable at infinity are not regular;  infinite derivations that cannot be built in finite time are
generated lazily, and only partial answers are shown.
More recently, a comprehensive theory has been proposed \cite{BasoldKL19} to provide operational semantics
that go beyond loop detection. 

Anyway, in coinductive logic programming, only standard coinduction is supported. The notion of \texttt{finally} clause, introduced in \cite{Ancona13}, allows the programmer to specify a fact to be resolved when a cycle is detected, instead of simply accepting the atom. The approach has been refined in \cite{AnconaDZ17}, following the guidelines given by the formal framework of generalized inference systems. That is, the programmer can write special clauses corresponding to corules, so that, when an atom is found for the second time, standard SLD resolution is triggered in the program enriched by the corules. \PB{However, this paradigm is very different from the object-oriented one, since based on relations rather than functions: cycles are detected on the same atom, where input and output are not distinguished, by unification.}

\smallskip
\noindent\textbf{Functional paradigm}
\emph{CoCaml} (\url{www.cs.cornell.edu/Projects/CoCaml}) \cite{Jeannin17,JeanninK12} is a fully-fledged extension of OCaml supporting non-well-founded data types and corecursive functions.
CoCaml, as OCaml, allows programmers to declare regular values through the let-rec construct\EZ{, and, moreover, detects cyclic calls as in our approach. However, whereas $\coFJ$ immediately evaluates the cyclic call by using the codefinition, the CoCaml approach is in two phases. First, a system of equations is constructed,
associating with each call a variable and partially evaluating the body
of functions, where calls are replaced with  associated variables. Then,} the system of equations is given to a \emph{solver}
specified in the function definition.
Solvers can be either pre-defined or written by the programmer in
order to enhance flexibility. \PB{An advantage that we see in our approach is that the programmer has to write the codefinition (standard code) rather than working at the meta-level to write a solver, which is in a sense a fragment of the interpreter.} \EZ{A precise comparison is difficult for the lack of a simple operational model of the CoCaml mechanism. In future work, we plan to develop such model, and to relate the two approaches on a formal basis.}

\smallskip
\noindent\textbf{Object-oriented paradigm}
A previous version of $\coFJ$ has been proposed in \cite{AnconaZ12}. At this time, however, the framework of inference systems with corules was still to come, so there was no formal model against which to check the given operational semantics, which, indeed, derived spurious results in some cases, as illustrated in \refToSection{operational} at page \pageref{sum-example}. 
The operational semantics provided in the current paper solves this problem, and is proved to be sound with respect to the abstract semantics. Moreover, we adopt a simpler representation of cyclic objects through capsules \cite{JeanninK12}.  A type system has been proposed
\cite{AnconaZ13} for the previous version of $\coFJ$ \EZ{to prevent \emph{unsafe} use of the ``undetermined'' value}. We leave to further work the investigation of typing issues for the approach presented in this paper.

\section{Conclusion}\label{sect:conclu}
The Java-like calculus presented in this paper promotes a novel programming style, which smoothly incorporates support for  
cyclic data structures and coinductive reasoning, in the object-oriented paradigm. 
\EZ{Our contribution is foundational: we provide an abstract semantics based on corules and show that it is possible to define a \emph{sound} operational model; such operational semantics is inductive, syntax-directed and deterministic, hence can be directly turned into an interpreter.
In order to get a ``real-world'' language, of course many other issues should be taken into account. }

\EZ{Our prototype implements the \emph{abstract} semantics on top of a Prolog meta-interpreter supporting flexible regular corecursion \cite{AnconaDZ17}. In this way, the inference system is naturally translated in Prolog\footnote{A logic program can be seen as an inference system where judgments are atoms.}, cyclic terms are natively supported, and their equality handled by unification. A fully-fledged interpreter of the \emph{operational} semantics should directly handle these issues and, moreover, attempt at some optimization.\EZComm{dire qualcosa di pi\`u?}}

\EZ{The current paper does not deal with types: an important concern is to guarantee \emph{type soundness}, statically ensuring that an undetermined value never occurs as receiver of field access or method invocation, as investigated in \cite{AnconaZ13} for the previous $\coFJ$ version \cite{AnconaZ12}.}

\EZ{Another issue is how to train developers to write codefinitions.  Standard recursion is non-trivial as well for beginners, whereas it becomes quite natural after understanding its mechanism. For regular corecursion the same holds, with is the additional difficulty of reasoning on infinite structures. Intuitively, the codefinition can be regarded as a base case to be applied when a loop is detected. Moreover, again as for standard recursion, this novel programming style could be integrated with proof techniques to show
the correctness of algorithms on cyclic data structures. Such proofs could be mechanized in proof assistants, as Agda, that provide built-in support for coinductive definitions and proofs by coinduction.}

\EZ{Finally, a non-trivial challenge is how to integrate regular corecursion, requiring to detect ``the same call'', with the notion of mutable state. Likely, some immutability constraints will be needed, or a variant of the model where such a check requires a stateless computation. Another solution is to consider the check as an assertion that can be disabled if the programmer has verified the correctness of the method by hand or assisted by a tool.}


The semantics of flexible regular corecursion in the paper is the operational counterpart of that obtained by considering recursive functions as relations, and recursive definitions (with codefinition) as inference systems (with corules). We prove that the operational semantics is \emph{sound} with respect to that interpretation. Obviously, \emph{completeness} does not hold in general, since the abstract semantics deals with  not only cyclic data structures (such as $[2,1]^\omega$), but arbitrary non-well-founded structures (such as the list of natural numbers). 
\EZ{Even considering only
regular proof trees in the abstract semantics, in some subtle cases there is more than one admissible result\footnote{\EZ{For instance, the list with no repetitions extracted from $[1,2]^\omega$ can be either $[1,2]$ or $[2,1]$. }}, whereas the operational semantics, being deterministic, finds ``the first'' among such results, as reasonable in an implementation\EZComm{la sem. intermedia invece dovrebbe trovarli tutti}. We plan to investigate such completeness issues in further work, also in the more general framework of inference systems, that is, to characterize judgments which have a regular proof tree.} 

We also plan to study how to deal with flexible corecursion in other programming paradigms, notably in the functional paradigm, and to compare on a formal basis this approach with the CoCaml approach relying on solvers, rather than codefinitions. 

As already discussed in the Introduction, lazy evaluation and regular corecursion are complementary approaches to deal with infinite data structures. With the lazy approach, arbitrary \EZ{(computable)} non-well-founded data structures are supported. However, we cannot compute results which need to explore the whole structure, whereas, with regular corecursion, this becomes possible for cyclic structures: for instance we can compute \texttt{allPos one\_two}, which diverges in Haskell. 
A natural question is then whether it is possible to extend the regular corecursion approach to manage also non-regular objects, thus overcoming the principal drawback with respect to the lazy approach. 
A possible interesting direction, exploiting the work of Courcelle \cite{Courcelle83} on infinite trees, could be to move from regular to \emph{algebraic} objects. 

\bibliography{bib}

\appendix
\section{Proofs}

\begin{proofOf}{\refToLemma{op-conservative-plus}}
The proof is by induction on the definition of $\opsem{\E}{\emptyset}{\callEnv}{\open}{\emptyset}$.\\
\begin{description}
\item[\refToRule{val}] 
By hypothesis $\Pair{\open}{\emptyset}$ is a capsule, then, 
since the environment is empty, $\open$ is a closed \FJ value, 
hence, we can prove by induction on $\open$ that ${\validInd{\isFJ}{\BigStepFJ{\val}{\val}}}$. 

\item[\refToRule{field}] 
By hypothesis we have $\opsem{\E}{\emptyset}{\callEnv}{\open}{\emptyset}$, with $\unf(\open, \emptyset) = \open = \NewExpr{\C}{\open_1,\ldots,\open_n}$ and $n\ge 1$, then we get $\validInd{\isFJ}{\BigStepFJ{\E}{\open}}$ by induction hypothesis. 
Hence we get the thesis by rule \refToRule{\FJ-field}. 

\item[\refToRule{new}] 
By hypothesis, we have 
$\opsem{\E_i}{\emptyset}{\callEnv}{\open_i}{\mapEnv'_i}$, for all $i \in 1..n$.
Since $\bigmapUnion{i\in 1..n}{\mapEnv'_i} = \emptyset$, we have $\mapEnv'_i = \emptyset$ for all $i \in 1..n$, hence, by induction hypothesis, we get 
$\validInd{\isFJ}{\BigStepFJ{\E_i}{\open_i}}$.
Then, we get the thesis by rule \refToRule{\FJ-new}.

\item[\refToRule{invk-ok}] 
By hypothesis we have  $\opsem{\E_i}{\emptyset}{\callEnv}{\open_i}{\mapEnv'_i}$ for all $i\in 0..n$, and $\opsem{\Subst{\Subst{\E}{\open_0}{\kwThis}}{\openBar}{\xBar}}{\mapEnvU}{\UpdateEnvCo{\callEnv}{\call}{\x}}{\open}{\emptyset}$. 
Then, by \refToLemma{env_contained}, we get $\mapEnvU \leq \emptyset$, thus $\mapEnvU = \emptyset$, and, since $\mapEnv'_i\leq \mapEnvU$, for all $i \in 0..n$, by definition, we also get $\mapEnv'_i = \emptyset$ for all $i \in 0..n$. 
Therefore, by induction hypothesis, we get $\validInd{\isFJ}{\BigStepFJ{\E_i}{\open_i}}$, for all $i \in 0..n$, and $\validInd{\isFJ}{\BigStepFJ{\Subst{\Subst{\E}{\open_0}{\kwThis}}{\openBar}{\xBar}}{\open}}$, 
hence, we get the thesis by rule \refToRule{\FJ-invk}. 

\item[\refToRule{invk-check}] This case is empty since to apply the rule it should be $\UpdateEnv{\mapEnvPrime}{\x}{\open} \neq\emptyset$.
\item[\refToRule{corec}] This case is empty since to apply the rule it should be $\UpdateEnv{\mapEnvPrime}{\x}{\x} \neq\emptyset$.
\item[\refToRule{look-up}] This case is empty since to apply the rule it should be $\mapEnv\neq\emptyset$, while $\x\in\dom{\mapEnv}$. 
\end{description}
\end{proofOf}

\begin{proofOf}{\refToProp{solutions}}\hspace*{\fill}
\begin{enumerate}
\item If $\undef{\unf(\open,\mapEnv)}$, then trivially $\open = \x$ and $\undef{\unf(\mapEnv(\x),\mapEnv)}$. 
Hence, by definition of $\varEquiv{\mapEnv}{}{}$, we get the thesis. 
\item If $\x$ is a free variable in $\expCaps{\open}{\mapEnv}$, then by definition of tree expansion we have $\undef{\unf(\x,\mapEnv)}$, hence the thesis follows by \ref{prop:solutions:1}. 
\item It is enough to show that, if $\mapEnv(x) = y$, then $\subst(x) = \subst(y)$, and this is true as 
$\subst(x) = \ApplySubst{\mapEnv(x)}{\subst} = \subst(y)$, since $\subst \in \solutions(\mapEnv)$. 
\item We know that either $\open = \NewExpr{\C}{\open_1,\ldots,\open_n}$ or $\open = \x$ and $\mapEnv^n(\x) = \NewExpr{\C}{\open_1,\ldots,\open_n}$. 
In the former case, the thesis is immediate. 
In the latter, it follows by induction on $n$, since $\ApplySubst{\x}{\subst} = \subst(\x) = \ApplySubst{\mapEnv(\x)}{\subst}$, as $\subst \in \solutions(\mapEnv)$. 
\item The proof is by coinduction. 
If $\undef{\unf(\open,\mapEnv)}$, then $\open = \x \in \dom{\mapEnv}$, $\expCaps{\x}{\mapEnv} = \x$ hance the thesis is immeidate. 
If $\unf(\open,\mapEnv) = \NewExpr{\C}{\open_1,\ldots,\open_n}$, then $\expCaps{\open}{\mapEnv} = \NewExpr{\C}{\expCaps{\open_1}{\mapEnv}, \ldots, \expCaps{\open_n}{\mapEnv}}$.
If $\ApplySubst{\open}{\subst} = \ApplySubst{\expCaps{\open}{\mapEnv}}{\subst}$, since by point \ref{prop:solutions:4} we have $\ApplySubst{\open}{\subst} = \NewExpr{\C}{\ApplySubst{\open_1}{\subst},\ldots,\ApplySubst{\open_n}{\subst}}$ and, by definition, $\ApplySubst{\expCaps{\open}{\mapEnv}}{\subst} = \NewExpr{\C}{\ApplySubst{\expCaps{\open_1}{\mapEnv}}{\subst},\ldots,\ApplySubst{\expCaps{\open_n}{\mapEnv}}{\subst}}$, hence we get $\ApplySubst{\open_i}{\subst} = \ApplySubst{\expCaps{\open_i}{\mapEnv}}{\subst}$ for all $i \in 1..n$, as needed. 
\end{enumerate}
\end{proofOf}

\begin{proofOf}{\refToLemma{tree-sem-sound}}
We first prove that $\semCaps{\open_2}{\mapEnv_2}  \subseteq \semCaps{\open_1}{\mapEnv_1}$. 
Let $\subst_2 \in \solutions(\mapEnv_2$. 
We have to find $\subst_1 \in \solutions(\mapEnv_1)$ such that $\ApplySubst{\open_2}{\subst_2} = \ApplySubst{\open_1}{\subst_1}$. 
We define a regular system of equations, with variables in $\dom{\mapEnv_1}$, which is finite by definition,  as follows: 
\[
s(x) = \begin{cases}
\subst_2(y)  &  \exists x'\in \Undetermined{\mapEnv_1}. \varEquiv{\mapEnv_1}{x}{x'} \text{ and } x'\varRel y \\
\mapEnv_1(x) & \text{otherwise}
\end{cases}
\]
The system is well-defined because, if there are $x_1,x_2 \in \Undetermined{\mapEnv_1}$ such that $\varEquiv{\mapEnv_1}{x}{x_1}$, $\varEquiv{\mapEnv_1}{x}{x_2}$,  $x_1\varRel y_1$ and $x_2\varRel y_2$, then $\varEquiv{\mapEnv_1}{x_1}{x_2}$, since $\varEquiv{\mapEnv_1}{}{}$ is an equivalence relaition, thus $\varEquiv{\mapEnv_2}{y_1}{y_2}$, as $\varRel$ is a $\mapEnv_1,\mapEnv_2$-renaming; 
hence, by \refToPropItem{solutions}{3}, we have $\subst_2(y_1) = \subst_2(y_2)$.
By results in \cite{Courcelle83}, there exists a solution $\subst_1$ of $s$, we now show that $\subst_1 \in \solutions(\mapEnv_1)$, that is, for all $x \in \dom{\mapEnv_1}$, $\subst_1(x) = \ApplySubst{\mapEnv_1(x)}{\subst_1}$. 
There are two cases:
if there is $x'\in\Undetermined{\mapEnv_1}$ such that $\varEquiv{\mapEnv_1}{x}{x'}$ and $x'\varRel y$, then 
$\undef{\unf(x,\mapEnv_1)}$ and 
$\undef{\unf(\mapEnv_1(x), \mapEnv_1)}$, hence $\mapEnv_1(x)$ is a variable, and,
since $\varEquiv{\mapEnv_1}{x}{\mapEnv_1(x)}$,  we get $\varEquiv{\mapEnv_1}{\mapEnv_1(x)}{x'}$, hence, by definition of $s$,   $\subst_1(x) = \subst_2(y) = \subst_1(\mapEnv_1(x)) = \ApplySubst{\mapEnv_1(x)}{\subst_1}$;
otherwise, 
we have  $\subst_1(x) = \ApplySubst{s(x)}{\subst_1} = \ApplySubst{\mapEnv_1(x)}{\subst_1}$, as needed.

Finally, we can prove by coinduction that, for all open  values $\open'_1$ and $\open'_2$ such that $\FV{\open'_1}\subseteq\dom{\mapEnv_1}$ and $\FV{\open'_2}\subseteq \dom{\mapEnv_2}$, if $\expCaps{\open'_1}{\mapEnv_1} =_\varRel \expCaps{\open'_2}{\mapEnv_2}$ then $\ApplySubst{\open'_1}{\subst_1} = \ApplySubst{\open'_2}{\subst_2}$. 
There are two cases:
\begin{itemize}
\item if $\undef{\unf(\open'_1,\mapEnv_1)}$, then $\open'_1 {=} x$, $\open'_2 = y$, $\undef{\unf(y, \mapEnv_2)}$ and $x\varRel y$, hence the thesis follows by construction of $\subst_1$;
\item if $\unf(\open'_1,\mapEnv_1) {=} \NewExpr{\C}{\open_{1,1},\ldots,\open_{1,n}}$, then $\expCaps{\open'_1}{\mapEnv_1} = \NewExpr{\C}{\expCaps{\open_{1,1}}{\mapEnv_1},\ldots,\expCaps{\open_{1,n}}{\mapEnv_1}}$ and $\expCaps{\open'_2}{\mapEnv_2} = \NewExpr{\C}{\expCaps{\open_{2,1}}{\mapEnv_2},\ldots,\expCaps{\open_{2,n}}{\mapEnv_2}}$, 
hence, for all $1\in 1..n$ we have $\expCaps{\open_{1,i}}{\mapEnv_1} =_R \expCaps{\open_{2,i}}{\mapEnv_2}$, then the thesis follows by coinductive hypothesis. 
\end{itemize}
Therefore, this proves that $\semCaps{\open_2}{\mapEnv_2} \subseteq \semCaps{\open_1}{\mapEnv_1}$. 
To get the other inclusion, it is enough to note that, if $\expCaps{\open_1}{\mapEnv_1} =_\varRel \expCaps{\open_2}{\mapEnv_2}$, then 
$\expCaps{\open_2}{\mapEnv_2} =_{\varRel\op} \expCaps{\open_1}{\mapEnv_1}$, hence the thesis follows from what we have just proved. 
\end{proofOf}

\begin{proofOf}{\refToProp{sem-eq}} \hspace*{\fill}
If $\undef{\unf(\open_1,\mapEnv_1)}$ and $\unf(\open_2,\mapEnv_2) = \NewExpr{\C}{\open'_1,\ldots,\open'_n}$, then $\open_1 = x$ for some undetermined variable $x \in \dom{\mapEnv_1}$ hence $\semCaps{\open_1}{\mapEnv_1} = \coFJaVals$. 
Instead, if $\val \in \semCaps{\open_2}{\mapEnv_2}$, we have $\val = \NewExpr{\C}{\val_1,\ldots, \val_n}$. 
Therefore, the value $\NewExpr{\C}{\val'_1,\ldots,\val'_{n+1}} \in \coFJaVals$ exists by definition of values, but it does not belong to $\semCaps{\open_2}{\mapEnv_2}$, which is a contradiction. 
This proves \ref{prop:sem-eq:1} and the first part of \ref{prop:sem-eq:2}. 
To conclude, assume $\unf(\open_1,\mapEnv_1) = \NewExpr{\C}{\open_{1,1},\ldots, \open_{1,n}}$ and $\unf(\open_2,\mapEnv_2) = \NewExpr{\C}{\open_{2,1},\ldots,\open_{2,n}}$, and consider $i \in 1..n$. 
Consider a solution $\subst_1 \in \solutions(\mapEnv_1)$, then $\ApplySubst{\open_{1,i}}{\subst_1} \in \semCaps{\open_{1,i}}{\mapEnv_1}$, and we have to prove $\ApplySubst{\open_{1,i}}{\subst_1} \in \semCaps{\open_{2,i}}{\mapEnv_2}$. 
By hypothesis, we know that there is a solution $\subst_2 \in \solutions(\mapEnv_2)$ such that $\ApplySubst{\open_1}{\subst_1} = \ApplySubst{\open_2}{\subst_2}$, 
hence, by \refToPropItem{solutions}{4}, we get 
$\NewExpr{\C}{\ApplySubst{\open_{1,1}}{\subst_1},\ldots,\ApplySubst{\open_{1,n}}{\subst_1}} = \NewExpr{\C}{\ApplySubst{\C}{\open_{2,1}}{\subst_2},\ldots,\ApplySubst{\open_{2,n}}{\subst_2}}$, and in particular
$\ApplySubst{\open_{1,i}}{\subst_1} = \ApplySubst{\open_{2,i}}{\subst_2} \in \semCaps{\open_{2,i}}{\mapEnv_2}$, as needed. 
This proves $\semCaps{\open_{1,i}}{\mapEnv_1} \subseteq \semCaps{\open_{2,i}}{\mapEnv_2}$, the other inclusion is immeidate. 
\end{proofOf}

\begin{proofOf}{\refToLemma{tree-sem-complete}}
Following \cite{Courcelle83}, given an infinite open value $t$, each subtree of $t$ is identified by a finite sequence $\beta$ of positive natural numbers, and it will be denoted by $t(\beta)$. 
Extending the notation for capsules, we denote by $\semCaps{t}{\mapEnv}$ the set $\{ \ApplySubst{t}{\subst} \mid \subst \in \solutions(\mapEnv) \}$. 
Set $t_1 = \expCaps{\open_1}{\mapEnv_1}$ and $t_2 = \expCaps{\open_2}{\mapEnv_2}$, we prove that for all $\beta$, if $\semCaps{t_1(\beta)}{\mapEnv_1} = \semCaps{t_2(\beta)}{\mapEnv_2}$ then 
either $t_1(\beta)$ and $t_2(\beta)$ are both variables, 
or $t_1(\beta) = \NewExpr{\C}{t_1(\beta1),\ldots,t_1(\beta n)}$, $t_2(\beta) = \NewExpr{\C}{t_2(\beta 1),\ldots,t_2(\beta n)}$ and, for all $i \in 1..n$, $\semCaps{t_1(\beta i)}{\mapEnv_1} = \semCaps{t_2(\beta i)}{\mapEnv_2}$. 
The proof is by induction on $\beta$, since by hypothesis $\semCaps{t_1}{\mapEnv_1} = \semCaps{t_2}{\mapEnv_2}$ and using \refToProp{sem-eq}. 

We define a relation $\varRel \subseteq \FV{t_1} \times \FV{t_2}$ as follows: 
$x\varRel y$ iff there is a sequence $\beta$ such that $x = t_1(\beta)$ and $y = t_2(\beta)$. 
We verify it is a $\mapEnv_1,\mapEnv_2$-renaming. . 
First note that if $x\varRel y$, since $\semCaps{t_1}{\mapEnv_1} = \semCaps{t_2}{\mapEnv_2}$, for all $\subst_1\in \solutions(\mapEnv_1)$ and $\subst_2\in \solutions(\mapEnv_2)$, if $\ApplySubst{t_1}{\subst_1} = \ApplySubst{t_2}{\subst_2}$, then $\subst_1(x) = \subst_2(y)$. 
Then, if $x_1\varRel y_1$, $x_2 \varRel y_2$ and $\varEquiv{\mapEnv_1}{x_1}{x_2}$, we have to prove $\varEquiv{\mapEnv_2}{y_1}{y_2}$. 
By contradiction, if $\varEquiv{\mapEnv_2}{y_1}{y_2}$ does not hold, since both $y_1$ and $y_2$ are undetermined variables by \refToPropItem{solutions}{2}, there is a solution $\subst_2 \in \solutions(\mapEnv_2)$ such that $\subst_2(y_1) \ne \subst_2(y_2)$.
Because $\semCaps{t_1}{\mapEnv_1} = \semCaps{t_2}{\mapEnv_2}$, there is a solution $\subst_1\in \solutions(\mapEnv_1)$ such that $\ApplySubst{t_1}{\subst_1} = \ApplySubst{t_2}{\subst_2}$, hence, since $x_1\varRel y_1$ and $x_2\varRel y_2$ we get, by definition of $\varRel$, $\subst_1(x_1) = \subst_2(y_1)$ and $\subst_1(x_2) = \subst_2(y_2)$. 
Now, since $\varEquiv{\mapEnv_1}{x_1}{x_2}$, by \refToPropItem{solutions}{3}, we have $\subst_1(x_1) = \subst_1(x_2)$, hence $\subst_2(y_1) = \subst_2(y_2)$, which is a contradiction. 
The proof of the other direction is the same. 

Finally, we can prove by coinduction that, for all $\beta$, $t_1(\beta) =_\varRel t_2(\beta)$, which, in particular, implies the thesis.
\end{proofOf}

\begin{proofOf}{\refToLemma{strict-bisim}}
We define systems of equations $s_2$ and $s_1$ as follows: 
\[
s_2(x) = \begin{cases}
\subst(x)    & x \in \dom{\mapEnv_2\cap\mapEnv_1} \\
\mapEnv_2(x) & \text{otherwise} 
\end{cases}
\]

\[
s_1(x) =\begin{cases}
\subst(x)    & x \in \dom{\mapEnv_1\cap\mapEnv_2} \\
\subst_2(y)  & x \notin \dom{\mapEnv_1\cap\mapEnv_2}\text{ and } \exists x'\in\dom{\mapEnv_1}.\ \varEquiv{\mapEnv_1}{x}{x'}\text{ and } x'\varRel y \\
\mapEnv_1(x) & \text{otherwise}
\end{cases}
\]
Both systems are well-defined since $\varRel$ is a $\mapEnv_2,\mapEnv_1$-renaming (see the proof of \refToLemma{tree-sem-sound} for details), and have solutions $\subst_2$ and $\subst_1$, respectively, by results in \cite{Courcelle83}. 

The fact that $\subst_2 \in \solutions(\mapEnv_2)$ is trivial. 
We now prove $\subst_1 \in \solutions(\mapEnv_1)$. 
We have three cases:
\begin{itemize}
\item If $x\in\dom{\mapEnv_1\cap\mapEnv_2}$, then $\ApplySubst{\mapEnv_1(x)}{\subst_1} = \ApplySubst{(\mapEnv_1\cap\mapEnv_2)(x)}{\subst_1}$ and, by definition, $\FV{(\mapEnv_1\cap\mapEnv_2)(x)} \subseteq \dom{\mapEnv_1\cap\mapEnv_2}$, hence we get 
$\ApplySubst{(\mapEnv_1\cap\mapEnv_2)(x)}{\subset_1} = \ApplySubst{(\mapEnv_1\cap\mapEnv_2)(x)}{\subst} = \subst(x)$;
thus $\ApplySubst{\mapEnv_1(x)}{\subst_1} = \subst(x) = \subst_1(x)$. 
\item If $x\notin \dom{\mapEnv_1\cap\mapEnv_2}$ and there is $x'\in \dom{\mapEnv_1}$ such that $\varEquiv{\mapEnv_1}{x}{x'}$ and $x'\varRel y$, then 
$\subst_1(x) = \subst_2(y)$, $\undef{\unf(\mapEnv_1(x),\mapEnv_1)}$ and $\varEquiv{\mapEnv_1}{x}{\mapEnv_1(x)}$, hence we have to prove $\subst_1(\mapEnv_1(x)) = \subst_2(y)$. 
Now, if $\mapEnv_1(x) \in \dom{\mapEnv_1\cap\mapEnv_2}$, then $\mapEnv_1(x) \in \dom{\mapEnv_2}$, hence, since $\varRel$ is strict, we have $\varEquiv{\mapEnv_2}{\mapEnv_1(x)}{y}$. 
Thus we get, by \refToPropItem{solutions}{3}, $\subst_2(y) = \subst_2(\mapEnv_1(x))$, and, by definition, we have $\subst_2(\mapEnv_1(x)) = \subst(\mapEnv_1(x)) = \subst_1(\mapEnv_1(x))$, hence $\subst_1(x) = \subst_2(y) = \subst_1(\mapEnv_1(x))$. 
\item Otherwise, we have $\subst_1(x) = \ApplySubst{s_1(x)}{\subst_1} = \ApplySubst{\mapEnv_1(x)}{\subst_1}$. 
\end{itemize}
Furthermore, we also have that, for all $x \in \dom{\mapEnv_1\cap\mapEnv_2}$, $<subst_1(x) = \subst(x) = \subst_2(x)$. 

We now prove, by coinduction, that if $\bisim[\varRel]{\Caps{\open'_1}{\mapEnv_1}}{\Caps{\open'_2}{\mapEnv_2}}$ then $\ApplySubst{\open'_1}{\subst_1} = \ApplySubst{\open'_2}{\subst_2}$, which in particular implies the thesis. 
If $\undef{\unf(\open'_1,\mapEnv_1)}$, then $\undef{\unf(\open'_2,\mapEnv_2)}$, $\open'_1 = x$, $\open'_2 = y$ and $x\varRel y$, hence, by definition, $\subst_1(x) = \subst_2(y)$, as needed. 
Otherwise, the thesis follows from \refToPropItem{solutions}{4} and coinduction hypothesis. 
\end{proofOf}

\begin{proofOf}{\refToLemma{env_contained}}
By induction on the definition of $\opsem{\E}{\mapEnv}{\callEnv}{\open}{\mapEnvPrime}$.\\ \PBComm{09/01: Updated}
\begin{description}
\item \refToRule{val}: We have that $\opsem{\open}{\mapEnv}{\callEnv}{\open}{{\mapEnv}}$. The thesis trivially holds since $\mapEnv\leq\mapEnv$.
\item  \refToRule{field}: By inductive hypothesis, $\opsem{\E}{\mapEnv}{\callEnv}{\open}{\mapEnvPrime}$ holds, and $\mapEnv\leq\mapEnvPrime$. Hence, $\opsem{\FieldAccess{\E}{\f}}{\mapEnv}{\callEnv}{\open_i}{\mapEnvPrime}$ holds as well, and $\mapEnv\leq\mapEnvPrime$. 
\item \refToRule{new}: By inductive hypothesis, $\opsem{\E_i}{\mapEnv}{\callEnv}{\open_i}{{\mapEnvPrime}_i}$ holds for all $i \in 1..n$, and $\mapEnv\leq\mapEnvPrime_i$ for all $i \in 1..n$. Hence, $\opsem{\NewExpr{\C}{\E_1,\ldots,\E_n}}{\mapEnv}{\callEnv}{\NewExpr{\C}{\open_1\ldots\open_n}}{\bigsqcup_{i \in 1..n} \mapEnvPrime_i}$ holds as well, and, since $\mapEnv\leq\mapEnvPrime_i$ for all $i \in 1..n$, we have that $\mapEnv\leq\bigsqcup_{i \in 1..n} \mapEnvPrime_i$.
\item \refToRule{invk-ok}: By inductive hypothesis, $\opsem{\E_i}{\mapEnv}{\callEnv}{\open_i}{{\mapEnvPrime}_i}$ holds for all $i \in 1..n$,\\ $\opsem{\Subst{\Subst{\E}{\open_0}{\kwThis}}{\openBar}{\xBar}}{\mapEnvU}{\UpdateEnvCo{\callEnv}{c}{x}}{\open}{\mapEnvPrime}$ holds, $\mapEnv\leq\mapEnvPrime_i$ for all $i \in 1..n$, and $\mapEnvU\leq\mapEnvPrime$. Hence, $\opsem{\MethodCall{\E_0}{\m}{\EBar}}{\mapEnv}{\callEnv}{\open}{\mapEnvPrime}$ holds as well and, since $\mapEnv\leq\mapEnvU$, and by the transitivity of the relation $\leq$, $\mapEnv\leq\mapEnvPrime$.
 \item \refToRule{invk-check}: By inductive hypothesis, $\opsem{\E_i}{\mapEnv}{\callEnv}{\open_i}{{\mapEnvPrime}_i}$ holds for all $i \in 1..n$,\\ 
$\opsem{\Subst{\Subst{\E}{\open_0}{\kwThis}}{\openBar}{\xBar}}{\mapEnvU}{\UpdateEnvCo{\callEnv}{c}{x}}{\open}{\mapEnvPrime}$ holds, 
$\opsem{ \Subst{\Subst{\E}{\open_0}{\kwThis}}{\openBar}{\xBar}}{\mapEnvU \sqcup\UpdateEnv{\mapEnv'}{x}{\open}}{\UpdateEnvCk{\callEnv}{c}{x}}{\open'}{\mapEnv''}$ holds, $\mapEnv\leq\mapEnvPrime_i$ for all $i \in 1..n$, $\mapEnvU\leq\mapEnvPrime$, and $\mapEnvU \sqcup\UpdateEnv{\mapEnv'}{x}{\open}\leq\mapEnv''$. Hence, $\opsem{\MethodCall{\E_0}{\m}{\EBar}}{\mapEnv}{\callEnv}{\x}{\UpdateEnv{\mapEnvPrime}{\x}{\open}}$ holds as well and, since $\mapEnv\leq\mapEnvU\leq\mapEnvPrime$, we have that $\mapEnv\leq\UpdateEnv{\mapEnvPrime}{\x}{\open}$.
 \item \refToRule{corec}: By inductive hypothesis, $\opsem{\E_i}{\mapEnv}{\callEnv}{\open_i}{{\mapEnvPrime}_i}$ holds for all $i \in 1..n$,\\ $\opsem{\Subst{\Subst{\Subst{\E'}{\open_0}{\kwThis}}{\openBar}{\xBar}}{x}{\Any}}{{\UpdateEnv{\mapEnvU}{\x}{\x}}}{\callEnv}{\open}{\mapEnvPrime}$ holds, $\mapEnv\leq\mapEnvPrime_i$ for all $i \in 1..n$, and $\UpdateEnv{\mapEnvU}{\x}{\x}\leq\mapEnvPrime$. Hence, $\opsem{\MethodCall{\E_0}{\m}{\E_1,\ldots,\E_n}}{\mapEnv}{\callEnv}{\open}{\mapEnvPrime}$ holds and, since $\mapEnv\leq\UpdateEnv{\mapEnvU}{\x}{\x}$, and by the transitivity of the relation $\leq$, $\mapEnv\leq{\UpdateEnv{\mapEnvPrime}{\x}{\x}}$.
 \item \refToRule{look-up}: By inductive hypothesis, $\opsem{\E_i}{\mapEnv}{\callEnv}{\open_i}{{\mapEnvPrime}_i}$ holds for all $i \in 1..n$, and $\mapEnv\leq\mapEnvPrime_i$ for all $i \in 1..n$. Hence, $\opsem{\MethodCall{\E_0}{\m}{\EBar}}{\mapEnv}{\callEnv}{\x}{ \mapEnvU }$ holds, and, $\mapEnv\leq\mapEnvU$.                                 
\end{description}
\end{proofOf}

\begin{proofOf}{\refToLemma{boundedness}}
By induction on the definition of $\intsem{\E}{\emptyset}{\callSet}{\val}{\callSetRes}$.
\begin{description}

\item[\refToRule{IN-val}] By hypothesis, we have that $\intsem{\val}{\emptyset}{\callSet}{\val}{\emptyset}$. The thesis is immediate by applying rule \refToRule{abs-co-val}.

\item[\refToRule{IN-field}] By hypothesis, we have that $\intsem{\E}{\emptyset}{\callSet}{\val}{\callSetRes}$. By inductive hypothesis, $\evalExtended{\E}{\val}$ holds. Thus, we can apply rule \refToRule{abs-field} and get the thesis.

\item[\refToRule{IN-new}] By hypothesis, we have that $\intsem{\E_i}{\emptyset}{\callSet}{\val_i}{\callSetRes_i}$ for all $i \in 1..n$. By inductive hypothesis, $\evalExtended{\E_i}{\val_i}$ holds for all $i \in 1..n$. Thus, we can apply rule \refToRule{abs-new} and get the thesis. 

\item[\refToRule{IN-invk-ok}-\refToRule{IN-invk-check}] By hypothesis, $\intsem{\E_i}{\emptyset}{\callSet}{\val_i}{\callSetRes_i}$ holds for all $i \in 0..n$, and \mbox{$\intsem{\Subst{\Subst{\E}{\val_0}{\kwThis}}{\vBar}{\xBar}}{\emptyset}{\callSet\cup\{\call\}}{\val}{\callSetRes}$} holds. By inductive hypothesis, $\evalExtended{\E_i}{\val_i}$ holds for all $i \in 0..n$ and and also $\evalExtended{\Subst{\Subst{\E}{\val_0}{\kwThis}}{\vBar}{\xBar}}{\val}$ holds. Thus, we can apply rule \refToRule{abs-invk} and get the thesis. Note that, in order to get the thesis, the third premise of rule \refToRule{invk-check} has not been used.

\item[\refToRule{IN-corec}] By hypothesis, $\intsem{\E_i}{\emptyset}{\callSet}{\val_i}{\callSetRes_i}$ holds for all $i \in 0..n$., and \\ $\intsem{\Subst{\Subst{\Subst{\E'}{\val_0}{\kwThis}}{\vBar}{\xBar}}{u}{\Any}}{\emptyset}{\callSet}{\val}{\callSetRes}$ holds. By inductive hypothesis, $\evalExtended{\E_i}{\val_i}$ holds for all $i\in 0..n$, and also  $\evalExtended{\Subst{\Subst{\Subst{\E'}{\val_0}{\kwThis}}{\vBar}{\xBar}}{u}{\Any}}{\val}$ holds. Thus, we can apply rule \refToRule{abs-co-invk} and get the thesis.

\item[\refToRule{IN-look-up}] This case is empty since to apply the rule it should be $\resEnv\neq\emptyset$.

\end{description}
\end{proofOf}

\EZ{\begin{lemma}\label{lemma:callset-check}
If $\intsem{\E}{\resEnv}{\callSet\cup\{\call\}}{\val}{\callSetRes}$, and $\intsem{\E}{\UpdateEnv{\resEnv}{\call}{\val}}{\callSet}{\val}{\callSetResCheck}$, then $\callSetResCheck\subseteq\callSetRes\setminus\{\call\}$. 
\end{lemma}
\begin{proof}
By induction on the definition of $\intsem{\E}{\resEnv}{\callSet\cup\{\call\}}{\val}{\callSetRes}$. 
\EZComm{nei casi della chiamata di metodo, ci vuole casistica sulle regole applicate per l'altro judgment}
\end{proof}}

\begin{proofOf}{\refToLemma{callSet}}
For brevity, we use $\RcallSet$ in place of $\callSet\cup\{\call\}$. In rules \refToRule{IN-invk-ok}, \refToRule{IN-invk-check}, \refToRule{IN-corec} and \refToRule{IN-look-up}, $\RcallSet$ abbreviates $\callSet\cup\{\call'\}$, so to distinguish between different calls that could be present in the call trace at the same time.

\noindent By induction on the definition of $\intsem{\E}{\resEnv}{\RcallSet}{\val}{\callSetRes}$.

\begin{description}

\item[\refToRule{IN-val}] By hypothesis, $\intsem{\val}{\resEnv}{\RcallSet}{\val}{\emptyset}$. We can trivially get the thesis by rule \refToRule{IN-val}

\item[\refToRule{IN-field}] By hypothesis, $\intsem{\E}{\resEnv}{\RcallSet}{\val}{\callSetRes}$, and $\call\not\in\callSetRes$. By inductive hypothesis, \mbox{$\intsem{\E}{\resEnv}{\callSet}{\val}{\callSetRes}$} holds. Thus, we can apply rule \refToRule{IN-field} and get the thesis.

\item[\refToRule{IN-new}] By hypothesis, $\intsem{\E_i}{\resEnv}{\RcallSet}{\val_i}{\callSetRes_i}$, and $\call\not\in\callSetRes_i$ for all $i\in 1..n$. By inductive hypothesis, $\intsem{\E_i}{\resEnv}{\callSet}{\val_i}{\callSetRes_i}$ holds for all $i\in 1..n$. Thus, we can apply rule \refToRule{IN-new} and get the thesis.

\item[\refToRule{IN-invk-ok}] By hypothesis, $\intsem{\E_i}{\resEnv}{\RcallSet}{\val_i}{\callSetRes_i}$, and $\call'\not\in\callSetRes_i$ for all $i\in 0..n$. Also by hypothesis, $\intsem{\Subst{\Subst{\E}{\val_0}{\kwThis}}{\vBar}{\xBar}}{\resEnv}{\RcallSet\cup\{\call\}}{\val}{\callSetRes}$ holds, and $\call'\not\in\callSetRes$. By inductive hypothesis, $\intsem{\E_i}{\resEnv}{\callSet}{\val_i}{\callSetRes_i}$ holds for all $i\in 0..n$ and $\intsem{\Subst{\Subst{\E}{\val_0}{\kwThis}}{\vBar}{\xBar}}{\resEnv}{\callSet\cup\{\call\}}{\val}{\callSetRes}$ holds. Thus, we can apply rule \refToRule{IN-invk-ok} and get the thesis.

\item[\refToRule{IN-invk-check}] \EZ{By hypothesis, $\intsem{\E_i}{\resEnv}{\RcallSet}{\val_i}{\callSetRes_i}$, and $\call'\not\in\callSetRes_i$, hence by inductive hypothesis, $\intsem{\E_i}{\resEnv}{\callSet}{\val_i}{\callSetRes_i}$ holds, for all $i\in 0..n$. Moreover, \mbox{$\intsem{\Subst{\Subst{\E}{\val_0}{\kwThis}}{\vBar}{\xBar}}{\resEnv}{\RcallSet\cup\{\call\}}{\val}{\callSetRes}$} holds with either $\call'\not\in\callSetRes$ or $\call=\call'$. Then, $\intsem{\Subst{\Subst{\E}{\val_0}{\kwThis}}{\vBar}{\xBar}}{\resEnv}{\callSet\cup\{\call\}}{\val}{\callSetRes}$ holds in the first case by inductive hypothesis, in the second case since \mbox{$\callSet\cup\{\call\}=\RcallSet\cup\{\call\}$}.\EZComm{se mettessi la side condition $\call\not\in\callSet$ questo secondo caso sarebbe escluso.}  Finally, \mbox{$\intsem{ \Subst{\Subst{\E}{\val_0}{\kwThis}}{\vBar}{\xBar}}{\UpdateEnv{\resEnv}{\call}{\val}}{\RcallSet}{\val}{\callSetResCheck}$} holds, and, from \refToLemma{callset-check},  $\call'\not\in\callSet''$, hence, by inductive hypothesis, $\intsem{ \Subst{\Subst{\E}{\val_0}{\kwThis}}{\vBar}{\xBar}}{\UpdateEnv{\resEnv}{\call}{\val}}{\callSet}{\val}{\callSetResCheck}$ hold. Thus, we can apply rule \refToRule{IN-invk-check} and get the thesis.}

\item[\refToRule{IN-corec}] By hypothesis, $\intsem{\E_i}{\resEnv}{\RcallSet}{\val_i}{\callSetRes_i}$ holds, and $\call'\not\in\callSetRes_i$ for all $i\in 0..n$. Also by hypothesis, $\intsem{\Subst{\Subst{\Subst{\E'}{\val_0}{\kwThis}}{\vBar}{\xBar}}{u}{\Any}}{\resEnv}{\RcallSet}{\val}{\callSetRes}$ holds and $\call'\not\in\callSetRes$. By inductive hypothesis, $\intsem{\E_i}{\resEnv}{\callSet}{\val_i}{\callSetRes_i}$ holds for all $i\in 0..n$ and $\intsem{\Subst{\Subst{\Subst{\E'}{\val_0}{\kwThis}}{\vBar}{\xBar}}{u}{\Any}}{\resEnv}{\callSet}{\val}{\emptyset}$ holds. Thus, we can apply rule \refToRule{IN-corec} and get the thesis.   

\item[\refToRule{IN-look-up}] By hypothesis, $\intsem{\E_i}{\resEnv}{\RcallSet}{\val_i}{\callSetRes_i}$ holds, and $\call'\not\in\callSetRes_i$ for all $i\in 1..n$. By inductive hypothesis, $\intsem{\E_i}{\resEnv}{\callSet}{\val_i}{\callSetRes_i}$ holds $\forall i \in 0..n$. Thus, we can apply rule \refToRule{IN-look-up} and get the thesis.
\end{description}
\end{proofOf}

\begin{lemma}\label{lemma:extEnvs}\EZComm{spostato qui perch\'e serve solo a provare il successivo}
 If $\intsem{\E}{\resEnv}{\callSet}{\val}{\callSetRes}$, then the following judgments hold: 
\begin{enumerate}
\item $\intsem{\E}{\resEnv}{\callSet\cup\{\call\}}{\val}{\callSetRes}$
\item $\intsem{\E}{\UpdateEnv{\resEnv}{\call}{\val'}}{\callSet}{\val}{\callSetRes}$ for any $\val'\in\coFJaVals$.
\end{enumerate}
\end{lemma}
\begin{proof}
The proof of both points is by induction on the definition of $\intsem{\E}{\resEnv}{\callSet}{\val}{\callSetRes}$. \EZComm{fare la prova; il secondo lemma nel caso generale \`e falso; entrambi questi lemmi servono solo nel caso $\callSetRes=\emptyset$, per\`o per fare la prova occorre generalizzare}
\end{proof}

\begin{proofOf}{\refToLemma{rho}}
In the proof of this lemma we will use $\resEnv'$ in place of $\UpdateEnv{\resEnv}{\call}{\val'}$. In rules \refToRule{IN-invk-ok}, \refToRule{IN-invk-check}, \refToRule{IN-corec} and \refToRule{IN-look-up}, $\resEnv'$ will be used in place of $\UpdateEnv{\resEnv}{\call'}{\val'}$ so to distinguish between different calls that could be present in the call trace at the same time.

\noindent We proceed by induction on the definition of $\intsem{\E}{\resEnv'}{\callSet}{\val}{\callSetRes}$.

\begin{description}

\item[\refToRule{IN-val}] By hypothesis, $\intsem{\val}{\resEnv'}{\callSet}{\val}{\emptyset}$. The thesis trivially holds by applying rule \refToRule{IN-val}.

\item[\refToRule{IN-field}] By hypothesis, $\intsem{\E}{\resEnv'}{\callSet}{\val}{\callSetRes}$. By inductive hypothesis, $\intsem{\E}{\resEnv}{\callSet}{\val}{\callSetRes}$ holds. Thus, we can apply rule \refToRule{IN-field} and get the thesis.

\item[\refToRule{IN-new}] By hypothesis, $\intsem{\E_i}{\resEnv'}{\callSet}{\val_i}{\callSetRes}$ holds for all $i\in 1..n$. By inductive hypothesis, $\intsem{\E_i}{\resEnv}{\callSet}{\val_i}{\callSetRes_i}$ holds for all $i\in 1..n$. Thus, we can apply rule \refToRule{IN-new} and get the thesis.

\item[\refToRule{IN-invk-ok}] By hypothesis, $\intsem{\E_i}{\resEnv'}{\callSet}{\val_i}{\callSetRes}$ holds for all $i\in 0..n$. Also by hypothesis, $\intsem{\Subst{\Subst{\E}{\val_0}{\kwThis}}{\vBar}{\xBar}}{\resEnv'}{\callSet\cup\{\call\}}{\val}{\callSetRes}$ holds. In order to use the inductive hypothesis, we apply \refToLemma{extEnvs} to the hypothesis $\intsem{\call'}{\resEnv}{\callSet}{\val'}{\emptyset}$ and obtain $\intsem{\call'}{\resEnv}{\callSet\cup\{\call\}}{\val'}{\emptyset}$. By inductive hypothesis, $\intsem{\E_i}{\resEnv}{\callSet}{\val_i}{\callSetRes_i}$ holds for all $i\in 1..n$, and $\intsem{\Subst{\Subst{\E}{\val_0}{\kwThis}}{\vBar}{\xBar}}{\resEnv}{\callSet\cup\{\call\}}{\val}{\callSetRes}$ holds. Thus, we can apply rule \refToRule{IN-invk-ok} and get the thesis.

\item[\refToRule{IN-invk-check}] By hypothesis, $\intsem{\E_i}{\resEnv'}{\callSet}{\val_i}{\callSetRes}$ holds for all $i\in 0..n$. Also by hypothesis, $\intsem{\Subst{\Subst{\E}{\val_0}{\kwThis}}{\vBar}{\xBar}}{\resEnv'}{\callSet\cup\{\call\}}{\val}{\callSetRes}$ and $\intsem{ \Subst{\Subst{\E}{\val_0}{\kwThis}}{\vBar}{\xBar}}{\UpdateEnv{\resEnv'}{\call}{\val}}{\callSet}{\val}{\callSetResCheck}$ hold. In order to use the inductive hypothesis, we apply \refToLemma{extEnvs} to the hypothesis $\intsem{\call'}{\resEnv}{\callSet}{\val'}{\emptyset}$ and obtain $\intsem{\call'}{\UpdateEnv{\resEnv}{\call}{\val}}{\callSet\cup\{\call\}}{\val'}{\emptyset}$. By inductive hypothesis, $\intsem{\E_i}{\resEnv}{\callSet}{\val_i}{\callSetRes_i}$ holds for all $i\in 1..n$, $\intsem{\Subst{\Subst{\E}{\val_0}{\kwThis}}{\vBar}{\xBar}}{\resEnv}{\callSet\cup\{\call\}}{\val}{\callSetRes}$ and $\intsem{ \Subst{\Subst{\E}{\val_0}{\kwThis}}{\vBar}{\xBar}}{\UpdateEnv{\resEnv}{\call}{\val}}{\callSet}{\val}{\callSetResCheck}$ hold. Thus, we can apply rule \refToRule{IN-invk-check} and get the thesis.

\item[\refToRule{IN-corec}] By hypothesis, $\intsem{\E_i}{\resEnv'}{\callSet}{\val_i}{\callSetRes}$ holds for all $i\in 0..n$. Also by hypothesis, $\intsem{\Subst{\Subst{\Subst{\E'}{\val_0}{\kwThis}}{\vBar}{\xBar}}{u}{\Any}}{\resEnv'}{\callSet}{\val}{\callSetRes}$ holds. By inductive hypothesis, $\intsem{\E_i}{\resEnv}{\callSet}{\val_i}{\callSetRes_i}$ holds for all $i\in 1..n$ and $\intsem{\Subst{\Subst{\Subst{\E'}{\val_0}{\kwThis}}{\vBar}{\xBar}}{u}{\Any}}{\resEnv}{\callSet}{\val}{\emptyset}$ holds. Thus, we can apply rule \refToRule{IN-corec} and get the thesis. 

\item[\refToRule{IN-look-up}] By hypothesis, $\intsem{\E_i}{\resEnv'}{\callSet}{\val_i}{\callSetRes}$ holds for all $i\in 0..n$. By inductive hypothesis, $\intsem{\E_i}{\resEnv}{\callSet}{\val_i}{\callSetRes_i}$ holds for all $i\in 1..n$. In order to proceed, we have to distinguish between two cases (recall that $\call'=\MethodCall{\val_0'}{\m'}{\vBar'}$ and $\call=\MethodCall{\val_0}{\m}{\vBar}$):
\begin{itemize}

\item If $\call'\neq\call$, then $\resEnv'(\call) = \resEnv(\call) = \val$, hence we get the thesis by applying rule \refToRule{look-up}. 

\item If $\call'=\call$, then, since $\intsem{\call'}{\resEnv}{\callSet}{\val'}{\callSetRes}$ holds by hypothesis, and $\intsem{\E_i}{\resEnv}{\callSet}{\val_i}{\callSetRes}$ holds by inductive hypothesis for all $i\in 0..n$, the judgment $\intsem{\MethodCall{\E_0}{\m}{\EBar}}{\resEnv}{\callSet}{\val'}{\callSetRes}$ can still be derived. \PBComm{Può andare bene come spiegazione?}

\end{itemize} 

\end{description}
\end{proofOf}

\begin{proofOf}{\refToTheorem{concrToInt}}
The proof is by induction on the definition of $\opsem{\E}{\mapEnv}{\callEnv}{\open}{\mapEnvPrime}$.
\begin{description}
\item[\refToRule{val}] Immediate by \refToRule{IN-val}, as $\callSetOp{\callEnv}{\mapEnv}{}\setminus\callSetOp{\callEnv}{\mapEnv}{} = \emptyset$.

\item[\refToRule{field}] 
By induction hypothesis, we have $\intsem{\ApplySubst{\E}{\subst}}{\resEnvOp{\callEnv}{\subst}}{\callSetTr{\callEnv}{\subst}}{\ApplySubst{\open}{\subst}}{\callSet}$, 
for some $\callSetDif{\callEnv}{\mapEnv}{\mapEnv'}{\subst} \subseteq \callSet \subseteq \callSetOp{\callEnv}{\mapEnv'}{\subst}$ and, 
by \refToPropItem{solutions}{4}, we get $\ApplySubst{\open}{\subst} = \NewExpr{\C}{\ApplySubst{\open_1}{\subst},\ldots,\ApplySubst{\open_n}{\subst}}$, hence the thesis follows by rule \refToRule{IN-field}.

\item[\refToRule{new}]  
Since $\mapEnv'_i \leq \bigmapUnion{i\in 1..n}{\mapEnv'_i}$, we have $\subst \in \solutions(\mapEnv'_i)$ by  \refToLemma{env-leq-sol}, hence, by induction hypothesis, we have
$\intsem{\ApplySubst{\E_i}{\subst}}{\resEnvOp{\callEnv}{\subst}}{\callSetTr{\callEnv}{\subst}}{\ApplySubst{\open_i}{\subst}}{\callSet_i}$, 
with $\callSetDif{\callEnv}{\mapEnv}{\mapEnv_i}{\subst} \subseteq \callSet_i \subseteq \callSetOp{\callEnv}{\mapEnv'_i}{\subst}$,  for all $i\in 1..n$. 
It is easy to see that $\callSetOp{\callEnv}{\bigmapUnion{i\in 1..n}{\mapEnv'_i}}{\subst} = \bigcup_{i\in 1..n} \callSetOp{\callEnv}{\mapEnv'_i}{\subst}$, 
then, if $\mapEnv' = \bigmapUnion{i\in 1..n}{\mapEnv'_i}$,  $\callSetDif{\callEnv}{\mapEnv}{\mapEnv'}{\subst}  \subseteq \bigcup_{i\in 1..n} \callSet_i \subseteq \callSetOp{\callEnv}{\mapEnv'}{\subst}$, 
hence the thesis follows by rule \refToRule{IN-new}. 

\item [\refToRule{invk-ok}] 
Since $\mapEnv'_i \leq \mapEnvU \leq \mapEnv'$, by \refToLemma{env-leq-sol}, we get $\subst \in \solutions(\mapEnv'_i)$, for all $i \in 1..n$. 
By induction hypothesis, we have 
$\intsem{\ApplySubst{\E_i}{\subst}}{\resEnvOp{\callEnv}{\subst}}{\callSetTr{\callEnv}{\subst}}{\ApplySubst{\open_i}{\subst}}{\callSet_i}$, 
with $\callSetDif{\callEnv}{\mapEnv}{\mapEnv'_i}{\subst} \subseteq \callSet_i \subseteq \callSetOp{\callEnv}{\mapEnv'_i}{\subst}$, for all $i \in 1..n$, and 
$\intsem{\ApplySubst{\Subst{\Subst{\E}{\open_0}{\kwThis}}{\openBar}{\xBar}}{\subst}}{\resEnvOp{\callEnv}{\subst}}{\callSetTr{\UpdateEnvCo{\callEnv}{\call}{\x}}{\subst}}{\ApplySubst{\open}{\subst}}{\callSet}$, 
with $\callSetDif{\UpdateEnvCo{\callEnv}{\call}{\x}}{\mapEnvU}{\mapEnv'}{\subst} \subseteq \callSet \subseteq \callSetOp{\UpdateEnvCo{\callEnv}{\call}{\x}}{\mapEnv'}{\subst}$.
Then we get $\callSetDif{\callEnv}{\mapEnv}{\mapEnv'}{\subst}  \subseteq \bigcup_{i\in 1..n} \callSet_i \cup \callSet \subseteq \callSetOp{\callEnv}{\mapEnv'}{\subst}$,
hence, since $\callSetTr{\UpdateEnv{\callEnv}{\call}{\x}}{\subst} = \callSetTr{\callEnv}{\subst} \cup \{\ApplySubst{\call}{\subst} \}$, 
in order to conclude by rule \refToRule{IN-invk-ok}, we have to prove that $\ApplySubst{\call}{\subst} \notin \callSet  \setminus \callSetTr{\callEnv}{\subst}$. 
Suppose $\ApplySubst{\call}{\subst} \in \callSet \setminus \callSetTr{\callEnv}{\subst} \subseteq \callSetOp{\UpdateEnvCo{\callEnv}{\call}{\x}}{\mapEnv'}{\subst} \setminus \callSetTr{\callEnv}{\subst}$, 
hence, there is $\call' \in \EZ{\dom{\mapEnv'\circ(\UpdateEnv{\callEnv}{\call}{\x})^{\neg\checklabel}}}$ such that $\ApplySubst{\call'}{\subst} = \ApplySubst{\call}{\subst}$ and $\ApplySubst{\call'}{\subst} \notin \callSetTr{\callEnv}{\subst}$. 
Since $\x \notin \dom{\mapEnv'}$, then $\call'\ne\call$, hence $\call'\in\EZ{\dom{\callEnvCo}}$, hence $\ApplySubst{\call'}{\subst} \in \callSetTr{\callEnv}{\subst}$, which is a contradiction. 

\item [\refToRule{invk-check}]
Since $\mapEnv'_i \leq \mapEnvU \leq \mapEnv'$, by \refToLemma{env-leq-sol}, we get $\subst \in \solutions(\mapEnv'_i)$, for all $i \in 1..n$. 
By induction hypothesis, we have 
$\intsem{\ApplySubst{\E_i}{\subst}}{\resEnvOp{\callEnv}{\subst}}{\callSetTr{\callEnv}{\subst}}{\ApplySubst{\open_i}{\subst}}{\callSet_i}$, 
with $\callSetDif{\callEnv}{\mapEnv}{\mapEnv'_i}{\subst}  \subseteq \callSet_i \subseteq \callSetOp{\callEnv}{\mapEnv'_i}{\subst}$, for all $i \in 1..n$, and 
$\intsem{\ApplySubst{\Subst{\Subst{\E}{\open_0}{\kwThis}}{\openBar}{\xBar}}{\subst}}{\resEnvOp{\callEnv}{\subst}}{\callSetTr{\UpdateEnvCo{\callEnv}{\call}{\x}}{\subst}}{\ApplySubst{\open}{\subst}}{\callSet}$, 
with $\callSetDif{\UpdateEnvCo{\callEnv}{\call}{\x}}{\mapEnvU}{\mapEnv'}{\subst} \subseteq \callSet \subseteq \callSetOp{\UpdateEnvCo{\callEnv}{\call}{\x}}{\mapEnv'}{\subst}$. 
Then we get $\callSetDif{\callEnv}{\mapEnv}{\mapEnv'}{\subst}  \subseteq \bigcup_{i\in 1..n} \callSet_i \cup \callSet \subseteq \callSetOp{\callEnv}{\mapEnv'}{\subst}$. 
Since $\bisim{\Caps{\x}{\UpdateEnv{\mapEnv'}{\x}{\open}}}{\Caps{\open'}{\mapEnv''}}$, by \refToCor{strict-bisim}, there is $\subst' \in \solutions(\mapEnv'')$ such that $\ApplySubst{\open'}{\subst'} = \subst(\x)$ and, for all $\x\in \dom{\mapEnv'}$, $\subst(\x) = \subst'(\x)$. 
Therefore, by induction hypothesis, we get 
$\intsem{\ApplySubst{\Subst{\Subst{\E}{\open_0}{\kwThis}}{\openBar}{\xBar}}{\subst'}}{\resEnvOp{\UpdateEnvCk{\callEnv}{\call}{\x}}{\subst'}}{\callSetTr{\UpdateEnvCk{\callEnv}{\call}{\x}}{\subst'}}{\ApplySubst{\open'}{\subst'}}{\callSet'}$, 
with $\callSetDif{\UpdateEnvCk{\callEnv}{\call}{\x}}{\mapEnv'}{\mapEnv''}{\subst}  \subseteq \callSet' \subseteq \callSetOp{\UpdateEnvCk{\callEnv}{\call}{\x}}{\mapEnv''}{\subst}$. 
Then, by construction of $\subst'$, we have $\ApplySubst{\Subst{\Subst{\E}{\open_0}{\kwThis}}{\openBar}{\xBar}}{\subst'} = \ApplySubst{\Subst{\Subst{\E}{\open_0}{\kwThis}}{\openBar}{\xBar}}{\subst}$ and $\resEnvOp{\UpdateEnvCk{\callEnv}{\call}{\x}}{\subst'} = \UpdateEnv{\resEnvOp{\callEnv}{\subst}}{\ApplySubst{\call}{\subst}}{\subst(\x)}$ and 
$\callSetTr{\UpdateEnvCk{\callEnv}{\call}{\x}}{\subst'} = \callSetTr{\callEnv}{\subst}$. 
Finally, since $\x \in \dom{\mapEnv'}$, $\call \in \callSetOp{\UpdateEnvCo{\callEnv}{\call}{\x}}{\mapEnv'}{}$, but $\call \notin \callSetOp{\UpdateEnvCo{\callEnv}{\call}{\x}}{\mapEnvU}{}$, since $\x$ is fresh, 
hence $\ApplySubst{\call}{\subst} \in \callSetDif{\UpdateEnvCo{\callEnv}{\call}{\x}}{\mapEnvU}{\mapEnv'}{\subst} \subseteq \callSet$,
thus the thesis follows by rule \refToRule{IN-invk-check}. 

\item [\refToRule{corec}] 
Since $\mapEnv'_i \leq \mapEnvU \leq \mapEnv'$, by \refToLemma{env-leq-sol}, we get $\subst \in \solutions(\mapEnv'_i)$, for all $i \in 1..n$. 
By induction hypothesis, we have 
$\intsem{\ApplySubst{\E_i}{\subst}}{\resEnvOp{\callEnv}{\subst}}{\callSetTr{\callEnv}{\subst}}{\ApplySubst{\open_i}{\subst}}{\callSet_i}$, 
with $\callSetDif{\callEnv}{\mapEnv}{\mapEnv'_i}{\subst}  \subseteq \callSet_i \subseteq \callSetOp{\callEnv}{\mapEnv'_i}{\subst}$, for all $i \in 1..n$, and 
$\intsem{\ApplySubst{\Subst{\Subst{\Subst{\E'}{\open_0}{\kwThis}}{\openBar}{\xBar}}{{\x}{\Any}}{\subst}}}{\resEnvOp{\callEnv}{\subst}}{\callSetTr{\callEnv}{\subst}}{\ApplySubst{\open}{\subst}}{\callSet}$, 
with $\callSetDif{\callEnv}{\UpdateEnv{\mapEnvU}{\x}{\x}}{\mapEnv'}{\subst} \subseteq \callSet \subseteq \callSetOp{\callEnv}{\mapEnv'}{\subst}$. 
Then we get $\callSetDif{\callEnv}{\mapEnv}{\mapEnv'}{\subst}  \subseteq \bigcup_{i\in 1..n} \callSet_i \cup \callSet \cup \{\ApplySubst{\call}{\subst}\}\subseteq \callSetOp{\callEnv}{\mapEnv'}{\subst}$, since $\x\in\dom{\mapEnv'}$, hence $\call \in \callSetOp{\callEnv}{\mapEnv'}{}$. 
In order to conclude by rule \refToRule{IN-corec}, we have to prove $\ApplySubst{\call}{\subst} \in \callSetTr{\callEnv}{\subst}$. 
We know that $\bisim{\Caps{\call}{\mapEnvU}}{\Caps{\call'}{\mapEnvU}}$ for some $\call' \in\EZ{ \dom{\callEnvCo}}$, hence, by \refToCor{strict-bisim}, we have $\ApplySubst{\call}{\subst} = \ApplySubst{\call'}{\subst}$, since $\subst \in \solutions(\mapEnvU)$. 
Therefore, $\ApplySubst{\call}{\subst} \in \callSetTr{\callEnv}{\subst}$, by definition. 

\item [\refToRule{look-up}]
Since $\mapEnv'_i \leq \mapEnvU $, by \refToLemma{env-leq-sol}, we get $\subst \in \solutions(\mapEnv'_i)$, for all $i \in 1..n$. 
By induction hypothesis, we have 
$\intsem{\ApplySubst{\E_i}{\subst}}{\resEnvOp{\callEnv}{\subst}}{\callSetTr{\callEnv}{\subst}}{\ApplySubst{\open_i}{\subst}}{\callSet_i}$, 
with $\callSetDif{\callEnv}{\mapEnv}{\mapEnv'_i}{\subst}  \subseteq \callSet_i \subseteq \callSetOp{\callEnv}{\mapEnv'_i}{\subst}$, for all $i \in 1..n$, and 
Then we get $\callSetDif{\callEnv}{\mapEnv}{\mapEnvU}{\subst} \subseteq \bigcup_{i\in 1..n} \callSet_i \subseteq \callSetOp{\callEnv}{\mapEnvU}{\subst}$.
In order to conclude by rule \refToRule{IN-look-up}, we have to prove that $\resEnvOp{\callEnv}{\subst}(\ApplySubst{\call}{\subst}) = \subst(\x)$. 
We know that there exists $\call'\in\EZ{\dom{\callEnvCk}}$ such that $\bisim{\Caps{\call}{\mapEnvU}}{\Caps{\call'}{\mapEnvU}}$ and $\callEnv(\call') = \x$, hence $\resEnvOp{\callEnv}{\subst}(\ApplySubst{\call}{\subst}) = \subst(\x)$, by definition. 

\end{description}
\end{proofOf}

\end{document}